\documentclass[twoside,11pt]{article}

\usepackage{blindtext}
\usepackage{amsmath,amsfonts}
\usepackage{algorithm,algorithmic}
\usepackage{booktabs,multirow}
\usepackage{xcolor}

\newtheorem{innercustomthm}{Theorem}
\newenvironment{customthm}[1]
{\renewcommand\theinnercustomthm{#1}\innercustomthm}{\endinnercustomthm}

\newtheorem{innercustomlemma}{Lemma}
\newenvironment{customlemma}[1]
{\renewcommand\theinnercustomlemma{#1}\innercustomlemma}{\endinnercustomlemma}

\newenvironment{claim}[1]{\par\noindent\underline{Claim:}\space#1}{}
\newenvironment{claimproof}[1]{\par\noindent\underline{Proof:}\space#1}{\hfill $\blacksquare$}

\usepackage[abbrvbib, preprint]{jmlr2e}

\usepackage{booktabs}
\usepackage{amsmath}
\usepackage{caption}
\usepackage{subcaption}
\usepackage{graphicx}

\usepackage{lastpage}

\ShortHeadings{Preprint}{Tang, Wang, Ma, Bogunovic, McAleer, and Yang}

\firstpageno{1}

\begin{document}

\title{Sample-Efficient Regret-Minimizing Double Oracle in Extensive-Form Games}

\author{\name Xiaohang Tang \email xiaohang.tang.20@ucl.ac.uk \\
       \addr 
       Department of Statistical Science,\\
       University College London.
       \AND
       \name Chiyuan Wang \email wang2021@stu.pku.edu.cn \\
       \addr
       YuanPei College, Peking University.
       \AND
      \name Chengdong Ma \email mcd1619@outlook.com \\
       \addr
       Institute for Artificial Intelligence, Peking University.
       \AND
       \name Ilija Bogunovic \email i.bogunovic@ucl.ac.uk \\
       \addr 
       Department of Electronic and Electrical Engineering, \\ University College London.
       \AND
      \name Stephen McAleer \email smcaleer@cs.cmu.edu \\
        \addr Carnegie Mellon University.\\  
      \name Yaodong Yang$^\dagger$ \email yaodong.yang@pku.edu.cn \\
             \addr
       Institute for Artificial Intelligence, Peking University. \\
       $^\dagger$ Corresponding author.
       }

\editor{}

\maketitle

\begin{abstract}
Extensive-Form Game (EFG) represents a fundamental model for analyzing sequential interactions among multiple agents and the primary challenge to solve it lies in mitigating sample complexity.
Existing research indicated that Double Oracle (DO) can reduce the sample complexity dependence on the information set number $|S|$ to the final restricted game size $X$ in solving EFG. 
This is attributed to the early convergence of full-game Nash Equilibrium (NE) through iteratively solving restricted games.
However, we prove that the state-of-the-art Extensive-Form Double Oracle (XDO) exhibits \textit{exponential} sample complexity of $X$, due to its exponentially increasing restricted game expansion frequency.
Here we introduce Adaptive Double Oracle (AdaDO) to significantly alleviate sample complexity to \textit{polynomial} by deploying the optimal expansion frequency.
Furthermore, to comprehensively study the principles and influencing factors underlying sample complexity, we introduce a novel theoretical framework Regret-Minimizing Double Oracle (RMDO) to provide directions for designing efficient DO algorithms.
Empirical results demonstrate that AdaDO attains the more superior approximation of NE with less sample complexity than the strong baselines including Linear CFR, MCCFR and existing DO. Importantly, combining RMDO with warm starting and stochastic regret minimization further improves convergence rate and scalability, thereby paving the way for addressing complex multi-agent tasks. \looseness=-1

\end{abstract}

\begin{keywords}
  Regret Minimization, Double Oracle, Sample Complexity, Game Theory, Nash Equilibrium
\end{keywords}

\section{Introduction}

Extensive-Form Game (EFG) is one of the widely studied fundamental models in game theory ~\citep{ritzberger2016theory}, and solving its Nash equilibrium (NE) is critical for addressing sequential decision-making problems constructed by multiple agents such as board games and auction bidding ~\citep{hart1992games}. Existing work have made strides in solving EFG, notably through a series of methods based on Counterfactual Regret Minimization (CFR) ~\citep{cfr,srm,mccfr}. These methods aim to approximate Nash equilibrium by traversing all nodes (information sets) within the game tree to compute counterfactual regrets and update strategies. However, it is evident that the sample complexity of CFR methods heavily depends on the number of information sets, denoted by $|S|$. Consequently, as the scale and complexity of the game increase, the resulting sample complexity by CFR methods becomes prohibitively high which significantly improves the intractability of solving EFGs.

To efficiently solve EFGs, prior work have attempted to introduce Double Oracle (DO) paradigm ~\citep{do,sdo,xdo,odo}, which has a superior mechanism with lower complexity dependencies than CFR family. 
The core idea of DO is to approximate NE only by resolving an expanding restricted game where players can only choose actions from a subset of the action space. The restricted game is expanded by adding the original game's Best Response (BR) against the NE in the restricted game (meta-NE). Since DO's restricted game typically halts its growth in the early stages before reaching the original game, DO can reduce the sample complexity dependence from $|S|$ to the final restricted game size $X$. 
The empirical results also demonstrate this advantage that the Extensive-Form DO (XDO) ~\citep{xdo} proposed based on DO framework can converge more efficiently to a less exploitable strategy than the regret minimization algorithm, and has become a state-of-the-art algorithm for solving EFGs.
However, XDO iteratively executing regret minimization in restricted games until the local exploitability reaches a threshold $\epsilon$, which then will be halved. This may lead to explosive growth in sample complexity in some cases but it is unclear how to mitigate it under complexity theory guidance.  This motivates the core question we aim to answer:

\begin{center}
    \textit{\textbf{Q: What causes high sample complexity and how to avoid them when designing more efficient extensive-form DO algorithms?}}
\end{center}

Firstly, we introduce a unified framework, \textbf{Regret-Minimizing Double Oracle (RMDO)} to theoretically understand the sources of sample complexity within the DO framework. RMDO is a generalization of existing DO methods including DO, XDO and ODO. We derive the sample complexity for RMDO framework to reach $\epsilon$-NE:
\begin{equation}
\tilde{\mathcal{O}}( k|A|X^3/\epsilon^2 + \sum_{j=1}^k|A|X^3/\epsilon^2m(j) +  Xm(j)),
\label{eq:rmdo_sc}
\end{equation}
where $j$ is the index of restricted game, $A$ is the action space, $k$ is the number of restricted games, $X$ is the largest game size among the games constructed by the support of NEs, and $m(\cdot)$ is the \textit{frequency function} of computing Best Response added to expand the restricted game. By setting different $m(\cdot)$, RMDO can be converted to existing existing DO methods. 
Based on RMDO, We have proved that even the state-of-the-art method XDO has the \emph{exponential} sample complexity in $k$, where $k$ represents the count of restricted games and is only bounded by the number of information sets in the final restricted game, denoted by $X$. This verified the concern about the complexity explosion of XDO.

Furthermore, based on the theoretical insights of RMDO, we propose an instance of RMDO called Adaptive Double Oracle (AdaDO), which employs the theoretically optimal frequency function to alleviate concerns about exponential sample complexity. 
AdaDO exhibits \textit{polynomial} sample complexity to reach $\epsilon$-NE, which matches the complexity lower bound of RMDO framework, and thus is more sample efficient than existing DO methods including XDO. Furthermore, to reduce the complexity caused by $k$ by integrating with warm starting for strategy initialization when solving a new restricted game, AdaDO demonstrates a significant improvement in the speed of exploitability decreasing empirically. We also adopt stochastic regret minimizer, exemplified by Monte-Carlo Counterfactual Regret Minimization (MCCFR) ~\citep{srm,mccfr}, for the restricted game solving, and manage to reduce the power of $X$ in the sample complexity of AdaDO and enhance the scalability of Double Oracle methods. We present a comprehensive summary of theoretical results in Table \ref{table:sc}, where Periodic Double Oracle (PDO) is naive improved instance by setting a constant expansion frequency but suffering from tuning this constant hyperparameter. Stochastic PDO (SPDO) and Stochastic Adaptive Double Oracle (SADO) are the natural extension of PDO and AdaDO adopting stochastic regret minimization for restricted game solving.

\begin{table*}[t!]
\centering
\caption{Main theoretical results of sample complexities for RMDO instances to reach $\epsilon$-NE in extensive-form games. We categorize the algorithms into reaching NE by regret minimization (RM) and stochastic regret minimization (SRM). Denote $|S|$ as the number of infosets. Since $X$ and $k$ are only bounded by $|S|$, here we display the degree of these two dominating terms $k$ and $X$ in the Sample Complexities. Besides, it is usually that $k \gg |A|$ in theory since $|A| \sim \mathcal{O}(|S|^{1/H})$, where $H$ is the horizon of the game, but $k$ is merely upper bounded by $X$. Exp. in the column of degree indicates that the complexity is exponential in the corresponding factor.}
\small\addtolength{\tabcolsep}{-2pt}
\begin{tabular}{@{}llccc@{}}
\toprule
Reach NE via & Algorithm & Sample Complexity & $k$ & $X$ \\ \midrule
\multirow{4}{*}{RM}
 & XODO ~\citep{odo} & $\tilde{\mathcal{O}}(k^2X^3/\epsilon^2)$ & 2 & 3 \\
 & XDO ~\citep{xdo} & $\tilde{\mathcal{O}}(k|A|X^3/\epsilon^2 + 4^k|A|X^3/\epsilon_0^2)$ & Exp. & 3 \\
 & \textbf{PDO} & $\tilde{\mathcal{O}}(k|A| X^3/\epsilon^2 )$ & 1 & 3 \\
 & \textbf{AdaDO} & $\tilde{\mathcal{O}}(k|A| X^3/\epsilon^2 )$ & 1 & 3 \\ 
 \midrule
\multirow{2}{*}{SRM} 
 & \textbf{SPDO} & $\tilde{\mathcal{O}}(   k|A|X^3/\epsilon^2 )$ & 1 & 3 \\
 & \textbf{SADO} & $\tilde{\mathcal{O}}( k|A|X^2/\epsilon^2)$ & 1 & 2 \\
\bottomrule
\end{tabular}
\label{table:sc}
\end{table*}


Empirical results in representative poker games and board game Sequential Blotto have demonstrated that AdaDO significantly outperforms XDO and can converge to exploitability solutions over $10$ times less exploitable than the strong regret minimization baseline, Linear Counterfactual Regret Minimization.
Notably, we observe a substantial improvement of in AdaDO with the warm starting technique, especially in a variant of Kuhn Poker, which enables DO to converge to up to $10^8$ times less exploitable solutions than DO without warm starting. In the setting of reaching NE via stochastic regret minimization, the instances of Stochastic RMDO can also generate a significantly less exploitable strategy compared to MCCFR. These results validate that AdaDO establishes a new state-of-the-art in EFG solving and provides a promising direction for addressing complex multi-agent decision-making tasks. \looseness=-1

\section{Preliminaries and Related Works}
\subsection{Two-Player Extensive-Form Games}
In this paper, our focus centers on Two-Player Zero-Sum Extensive-Form Games (EFGs) with perfect recall (all historical events can be remembered). EFGs are depicted using a game tree, wherein nodes correspond to players $i\in \mathcal{P}=\{1,2\}$. To model stochastic events, such as card dealing in Poker, Imperfect Information EFGs utilize a \textbf{Chance Player} denoted as $\mathbf{c}$. A crucial concept is the notion of a \textbf{History} ($h$), which represents a sequence of actions taken by the players and events uniquely attached to a node on the game tree such as dealt hand cards and public cards in Texas Hold'em Poker. For a given history $h$, $A(h)$ represents the set of available actions, and $P(h)$ denotes the player required to make a decision at $h$. Denote $h\cdot a$ as the history right after taking action $a$ at history $h$. Terminal histories, comprising the set $Z$, are linked to nodes where the game's payoff can be determined. The payoff at the terminal history $z\in Z$ is denoted as $v_i(z)$, representing the value for player $i$ at $z$, with the payoff range represented by $\Delta$.

Each player $i \in \mathcal{P}$ possesses an \textbf{Information Set} (Infoset/Infostate $s_i$), which corresponds to the set of indistinguishable histories from player $i$'s perspective. For example in poker, histories where opponent has different hand cards are indistinguishable to us. The set $S_i$ contains all the infosets where player $i$ must make decisions, and $S$ represents the union of information sets for all players: $S=\cup_{i\in \mathcal{P}}S_i$. The set $A(s_i)$ consists of all available actions for player $i$ at $s_i$. Player $i$'s \textbf{Strategy} is denoted by $\pi_i$, where $\pi_i(s_i,a)$ represents the probability of player $i$ taking action $a \in A(s_i)$ at the information set $s_i$. The joint strategy of both players is represented as $\pi=(\pi_1, \pi_2)$. The \textbf{Reaching probability} $x^{\pi}(h)$ signifies the probability of reaching history $h$ when players use joint strategy $\pi$. More specifically, $x^\pi(h) = \prod_{i\in \mathcal{P} \bigcup {\mathbf{c}}}x^\pi_i(h)$, where $x^\pi_i(h)=\prod_{h^{'}\cdot a \subset h} \pi_{\mathcal{P}(h^{'})}(h^{'},a)$ represents player $i$'s contribution.

Given a joint strategy $\pi=(\pi_1, \pi_2)$, the \textbf{expected value} of history $h$ for player $i$ is denoted as $v_i(h)$. If reaching probability $x^{\pi}(h)=0$, then $v_i^\pi(h)=0$; otherwise, we have:
\begin{align}
\label{ev}
v_i^\pi(h) =
\sum_{z\in Z} \frac{x^{\pi}(z)}{x^{\pi}(h)} v_i(z),\ x^{\pi}(h)>0 ,
\end{align}
where the value (utility) function of strategy $\pi$ is defined as:
\begin{align}
v_i(\pi_i, \pi_{-i}) = \sum_{z \in Z} v_i^{\pi}(z)x^\pi(z).
\end{align}

The \textbf{best response (BR)} of player $i$ against the strategy of the other player ($\pi_{-i}$) is denoted as $\mathbb{BR}_i(\pi_{-i})=\arg\max_{\pi_i}v_i(\pi_i, \pi_{-i})$. An $\epsilon$-\textbf{Nash Equilibrium (NE)} strategy $\pi^*$ satisfies the following conditions for every $i\in \mathcal{P}$:
\begin{align}
\min_{\pi_{-i}} v_i(\pi^*_i, \pi_{-i}) + \epsilon \geq v_i(\pi^*)\geq \max_{\pi_i} v_i(\pi_i, \pi^*_{-i}) - \epsilon.
\end{align}

Particularly, an exact NE is an $\epsilon$-NE when $\epsilon=0$. The \textbf{Exploitability} $e(\pi)$ is defined as the difference between the sum of values for each player when playing the best response ($\mathbb{BR}(\pi_{-i})$) and the value achieved when players playing $\pi$. Additionally, the \textbf{Support Size} of NE ($\pi^*$) refers to the number of actions that have positive probabilities at the infoset $s$ in NE $\pi^*$, denoted by $\text{supp}^{\pi^*}(s)$.

\subsection{Regret Minimization and Stochastic Regret Minimization}

Regret minimization methods approximate NE in EFGs if the methods have a sublinear regret upper bound or an average regret converging to zero. Given $\{\pi^t|\ t=1,\cdots,T\}$ is a sequence of strategies delivered by an algorithm, the cumulative regret of this algorithm is defined as: $R^T_i = \sum_{t=1}^T \max_{\pi} v_i(\pi, \pi_{-i}^t) - v_i(\pi_i^t, \pi_{-i}^t),$ and \emph{average} regret $\bar{R}^T_i = R^T_i/T$. Counterfactual Regret Minimization (CFR) ~\citep{cfr} is a regret minimization algorithm that minimize cumulative regret by minimizing counterfactual regret at each infoset via traversing the full game tree depth-firstly. CFR has been widely studied and used to develop superhuman Poker AI ~\citep{burch2012efficient,moravvcik2017deepstack,discfr,libratus,pluribus,brown2020combining}.

To compute the counterfactual regret, we calculate player i's expected values $v_i(\cdot)$ based on equation (\ref{ev}) and the instantaneous regrets at iteration $t\leq T$ of taking action $a$ in infoset $s$ \looseness=-1 
\begin{align}
    r_i^t(s,a) = \sum_{h \in s_i}x_{-i}^{\pi^t} (h) [v_i^{\pi^t}(h\cdot a) -  v_i^{\pi^t}(h)]
    \label{eq:instant_r}
\end{align}
The counterfactual regrets of CFR can be computed by uniform average of $r_i$ over all iterations: 
\begin{align}
    R_i^T(s,a) = \sum_{t=1}^T r_i^t(s,a) / T.
    \label{eq:cfr}
\end{align}
Denote $R_i^{t,+}(s,a)=\max\{0, R_i^{t}(s,a)\}$. In two-player zero-sum games, the regret of CFR is bounded by$\Delta|S_i|\sqrt{|A_i|T}$, if both players apply \textbf{regret matching}~~\citep{cfr} to strategy updates:
\begin{align}
    \pi_i^{t+1}(s, a) = \Bigg\{
    \begin{array}{cc}
         R_i^{t,+}(s,a)/ \sum_a R_i^{t,+}(s,a) ,  &\text{if} \sum_a R_i^{t,+}(s,a)>0; \\ 
         1/|A(s)| , &\text{else}
    \end{array}
    .
\label{eq:rm}
\end{align}

CFR needs to traverse the entire game tree to estimate $v_i^{\pi^t}(h)$ in equation \ref{eq:instant_r} and develop strategy for all states according to regret matching, which can be intractable in large games. Hence rather than computing exact regret, Stochastic Regret Minimization (SRM) ~\citep{srm} samples nodes for traversal to estimate regret and then minimize regret. The family of SRM, including outcome-sampling Monte-Carlo CFR and external-sampling Monte-Carlo CFR (MCCFR), have the following regret bound with at least probability $1-p$:
\begin{align}
    {R}^T_i \leq (1 + \sqrt{\frac{2}{p}})\frac{1}{\delta}\Delta|S_i|\sqrt{|A|T},
\end{align}
for any $p\in(0,1]$ and exploration parameter $\delta>0$. Specifically, $\delta=1$ in external-sampling MCCFR ~\citep{mccfr}. Based on stochastic regret minimiziation, previous works extend tabular CFR method to DeepCFR ~\citep{brown2019deep} using neural networks as function approximators ~\citep{xdo,dream,rebel}. MCCFR has the same magnitude of regret as CFR, but its sample complexity is much lower since the complexity of each iteration in CFR is $\mathcal{O}(|S|)$ while that of MCCFR is only $\mathcal{O}(H)$, where $H$ is the depth of the tree or horizon of the trajectory. In this way, when reaching $\epsilon$-NE, we bound the upper bound of regret to get the required iterations and multiply the required iterations with the complexity of each iteration to get the sample complexity. Then it can be easily derived that the sample complexity of MCCFR is reduced by up to $\tilde{\mathcal{O}}(|S|)$.



\subsection{Double Oracle Methods}

\begin{figure*}[t]
    \centering
    \includegraphics[width=1\textwidth]{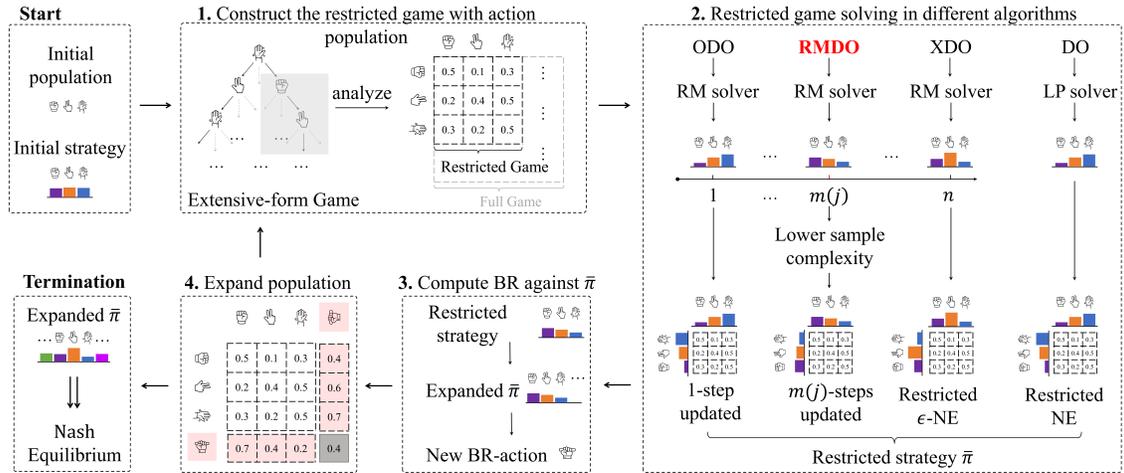}

    \caption{Flow chart of existing Double Oracle algorithms: Double Oracle (DO), Extensive-form Double Oracle (XDO), Online Double Oracle (ODO), and the method we proposed, namely Regret-Minimizing Double Oracle (RMDO).}
\label{fig:rmdo}
\end{figure*}

The Double Oracle (DO) technique ~\citep{do} originated as a method for addressing Normal-form Games (NFGs). At each time step $t$, DO maintains a population of pure strategies denoted as $\Pi_t$ for both players and proceeds to construct a restricted game by considering only the actions in $\Pi_t$. The Nash Equilibrium (NE) of this restricted game is then obtained through linear programming, and subsequently, the best response to this NE is incorporated into the population. This iterative process continues until the population reaches a stable state without further changes. Notably, DO offers a key advantage in solving a large game by focusing on solving relatively small restricted games ~\citep{do_small_sub1, do_small_sub2}.

To capitalize on the benefits offered by regret minimization methods, previous works have explored employing regret minimizers to update strategies within restricted games, rather than seeking to exactly compute a Nash Equilibrium (NE) strategy. Building on the ideas of Sequence-Form Double Oracle (SDO) ~\citep{sdo}, Extensive-Form Double Oracle (XDO) ~\citep{xdo} initializes a value $\epsilon$ as the threshold for stopping strategy updates in restricted games. It keeps executing regret minimization within the restricted game until achieving a local $\epsilon$-NE. Subsequently, $\epsilon$ is halved, and a new restricted game is constructed, where the best response against the current restricted-game NE and finally the strategy reset and repeat from the regret minimizing step. It is assured that the average strategy in the last restricted game converges to NE. In contrast, Online Double Oracle (ODO) ~\citep{odo} carries out only one iteration of regret minimization and then compute the best response against the restricted-game NE. If the best response actions are not in the restricted game, ODO will expand the restricted game with them and reset the strategy; Otherwise, repeat the previous steps. Notably, the average strategy of ODO is guaranteed to converge to NE. Despite these successes, the theoretical aspects concerning the convergence rate and sample complexity of DO methods in Extensive-Form Games have not been explored ~\citep{sdo}, making the DO algorithms hard to get further improvements. In this work, we propose a theoretical framework to study the theoretical problems.

Another blank of Double Oracle is the study on initialization for restricted game solving. Double Oracle resets the strategy each time it expands the population and constructs a new restricted game. Such a cold start to solving a new restricted game lacks the use of prior knowledge learned in previous restricted games, causing repeated training and inefficiency. Works have been done using warm starting to accelerate the convergence in PSRO, a DO-based multi-agent reinforcement learning method 
~\citep{lanctot2017unified,mcaleer2022anytime,smith2021iterative,zhou2022efficient}. However, PSRO is limited to normal-form games, while warm starting in extensive-form games is less predictable due to the sequential relation ships between information sets.\looseness=-1




\section{Regret-Minimizing Double Oracle}
\label{sec:rmdo}
In this section, we propose Regret-Minimizing Double Oracle (RMDO), a novel and versatile Double Oracle framework that integrates regret minimization to approximate the Nash Equilibrium of Extensive-Form Games (EFGs). To the best of our knowledge,  this study represents the first comprehensive analysis of convergence rate and sample complexity of regret-minimization-based Double Oracle for EFGs. 

RMDO consists of the same elements as the previous DO methods. Restricted game is constructed by considering only a subset of all pure strategies (actions). Population $\Pi_t$ containing the available pure strategies in the restricted game. Time window $T_j$, defined as a partition of the set of all iterations where the populations are the same: $\forall t_0,t_1 \in T_j, \Pi_{t_0}=\Pi_{t_1}$, plays a crucial role in RMDO and contributes to making it a generic framework. In contrast to existing DO methods, RMDO has the ability to expand the restricted game at any time due to the following key component:
\begin{definition}[Frequency Function] Denote the number of time windows from iteration $t=0$ to $T$ as $k$. $m(j)$ is defined as a mapping: $\mathcal{N} \cap [0,k-1] \rightarrow \mathcal{N}^{+}$. $m(j)$ represents the frequency of computing Best Response in time window $T_j$.
\end{definition}
Since the process of DO based on regret minimization is exactly to take turns to do regret minimization and compute the best response, we consider using $m(\cdot)$ to balance regret minimization and best response computation. Such balance is critical for DO methods to achieve a rapid convergence.

Presented in Algorithm \ref{alg:RMDO}, the formal RMDO procedure is as follows. At each iteration $t$, assuming the current time window is $j$, the restricted game $\mathbf{G}_t$ is constructed by restricting the pure strategies in the population $\Pi_t$ for players. Within $\mathbf{G}_t$, regret minimization is conducted by traversing the game tree, computing the regret of each infoset (node), and updating the strategy using any Counterfactual Regret Minimization (CFR) algorithm. At the outset of the procedure, when $t=0$, the construction of the restricted game and the strategy update are bypassed since $\Pi_0$ is empty. The expected value at $t=0$ is computed based on the joint strategy $\pi$ following a uniformly random policy. As the procedure progresses, when $t>0$ and the current time window is $T_j$, the joint average strategy of current window $\bar{\pi}=(\bar{\pi}_1, \bar{\pi}_2)$ is expanded to the original game every $m(j)$ iteration by setting the probabilities of actions not in the population to zero. Then the original game best response (BR), considering all actions in the original game, is computed against the expanded current-window average strategy, which is $\mathbf{a}_i^t=\arg\max_{\pi_i \in \Pi}v_i(\pi_i, \pi_{-i})$, for both players. $\mathbf{a}_i^t$ for $i=1,2$ are both merged to the population $\Pi_{t+1}$. Finally, if the population changes ($\Pi_{t+1} \neq \Pi_t$), a new time window is initiated, and $\pi_i^{t+1}$ is reset to a uniform random strategy. \looseness=-1

\begin{algorithm}[t!]
\begin{algorithmic}
    \caption{Regret-Minimizing Double Oracle}
    \label{alg:RMDO}
    \STATE \textbf{Input:} Frequency function $m(\cdot)$, window index $j=0$, uniform random strategy $\pi^0$.
    \STATE Set population $\Pi_1=\mathbb{BR}_i(\pi^0)$ for $i \in \{1,2\}$.
    \STATE Construct restricted game $\mathbf{G}_1$ with $\Pi_1$.
    \FOR{$t=1,\cdots,\infty$}
    \STATE Run one iteration of CFR in $\mathbf{G}_t$.
    \IF{$t \mod m(j)=0$}
    \STATE Compute average strategy $\Tilde{\pi}_i^t=\sum_{t\in T_j}\pi^t/|T_j|$.
    \STATE
    $\Pi_{t+1} = \Pi_t \cup  \mathbf{BR}_i(\Tilde{\pi}^t_{-i})$ for $i \in \{1,2\}$.
    \IF{$\Pi_{t+1} \neq \Pi_{t}$}
    \STATE Start new window: $j = j+1$.
    \STATE Reset strategy $\pi^{t+1}$.
    \STATE Construct restricted game $\mathbf{G}_{t+1}$ with $\Pi_{t+1}$.
    \ENDIF
    \ENDIF
    \ENDFOR
\end{algorithmic}
\end{algorithm}

\subsection{Convergence Guarantee and Sample Complexity}
\label{sec:tighter_sc}
Then we investigate the convergence guarantee and sample complexity of RMDO. We first define an important statistic related to the support size of Nash Equilibrium and prove that $k$ is bounded by this statistic. This is the key to the convergence guarantee. The reason is as following. Regret minimization algorithms can converge to $\epsilon$-NE by iteratively updating strategy in a static game. But in RMDO, a regret minimizer is employed in the restricted game expanding over time. Thus if the restricted game stops expanding at some finite iteration, the convergence of RMDO is guaranteed. The following lemma proves that the number of restricted games is finite regardless the final iteration $T$.
\begin{definition}
\label{def:x}
Denote the number of time windows from iteration
$t = 0$ to $T$ as $k$. Define 
\begin{align}
    X= \sum_i |S_{i,k}|,
\end{align}
as the size of the final restricted game. 
\end{definition}
We then show that the statistics above are bounded:
\begin{lemma}
\label{lemma:x}
$\min_{\pi\in \Pi^*}\max_{s\in S}\text{supp}^{\pi}(s) < k \leq X \leq |S| =\sum_i |S_i|$, where $k$ is the number of restricted game during the whole process of Double Oracle.
\end{lemma}
RMDO will converge by doing regret minimization in the final restricted game, after $k-1$ times of restricted game expanding. Then $X$, standing for the number of infosets in the largest one of these games, can be small in small support games ~\citep{bovsansky2016algorithms}. In practice, $X$ and $k$ are usually strictly less than $|S|$, but $X$ is usually larger than $k$. For example, in Hold'em poker games, at the infostate where we have the weakest hand cards but are facing with a large value bet by the opponents, the equilibrium strategy will always take the action of fold regardless of any other actions. Such appearance of pure strategy in specific states can significantly reduce $X$ by the number of nodes exponential in the horizon, and thus lead to $X<|S|$.

Now we investigate the sample complexity of RMDO. there are two kinds of average strategies produced by RMDO. \textbf{Overall average strategy (OAS)} is the average strategy over all time windows
\begin{equation}
    \bar{\pi}^{t} = \sum_{t=0}^T\pi^t/T.
\end{equation}
\textbf{Last-window average strategy (LAS)} is the average strategy in the final time window
\begin{equation}
    \tilde{\pi}^{t} = \sum_{t\in T_k}\pi^t/|T_k|,
\end{equation}
where $k$ is in hindsight the number of time windows (i.e. the number of restricted games) of training till reaching the $\epsilon$-NE. In the empirical study, OAS is much worse than LAS in terms of the decreasing speed of exploitability during training ~\citep{tang2023regret}. Besides, OAS is only used in Online Double Oracle, and not used for developing new methods in this paper, we put the regret bound analysis of OAS in Appendix (Theorem \ref{theo:RMDO_oas_bound}). Here we only introduce the sample complexity of LAS:

\begin{theorem}
\label{theo:rmdo_sc}
The sample complexity of LAS of RMDO to reach $\epsilon$-NE is:
\begin{align}
\tilde{\mathcal{O}}( k|A|X^3/\epsilon^2 + \sum_{j=1}^k|A|X^3/\epsilon^2m(j) +  Xm(j))
\label{eq:rmdo_sc}
\end{align}
\end{theorem}
Based on Theorem \ref{theo:rmdo_sc}, we can derive the sample complexity of existing DO methods, and accordingly propose more sample-efficient RMDO instances.


\subsection{Existing Frequency Schemes}
\label{sec:existing_do}
The complexity of approximating NE using RMDO is influenced by the choice of frequency function $m(j)$ for best response computation, as demonstrated in the previous section through theoretical analysis. For analysis on existing DO methods, we present various existing RMDO instantiations in this section.

\subsubsection{Online Double Oracle for Extensive-Form Games}

We introduce an extension of the Online Double Oracle (ODO) algorithm called Extensive-Form Online Double Oracle (XODO). This algorithm integrates the Sequence-form Double Oracle (DO) framework with Counterfactual Regret Minimization (CFR) to effectively address extensive-form games. The construction and updating of the restricted game and strategy in XODO closely resemble the DO framework employed in ODO. In each iteration, XODO extends the restricted game by computing the best response against the average strategy within the current window. Notably, as XODO computes the best response in each iteration following regret minimization, it is equivalent to the Regret-Minimizing Double Oracle (RMDO) algorithm with $m(j)=1$.

\begin{proposition}
\label{theo:xodo_bound} 
XODO is an instance of RMDO when $m(\cdot) \equiv 1$, thus given the regret minimizer with $\tilde{\mathcal{O}}(|S_i|\sqrt{|A|T})$ regret upper bound, the sample complexity to reach $\epsilon$-NE is 
\begin{align}
\tilde{\mathcal{O}}(2k^2X^3/\epsilon^2).
\end{align}
\end{proposition}

\subsubsection{Extensive-Form Double Oracle}

The Extensive-form Double Oracle (XDO) algorithm is initialized with a given threshold $\epsilon_0$, which is divided by two each time the local exploitability of the regret minimizer meets the threshold. The local exploitability is the exploitability in the restricted game. In time window $T_j$, the algorithm performs regret minimization for more than $4^j|S_{i,j}|^2|A_{i,j}|/\epsilon_0^2$ iterations before computing the best response. Here $A_{i,j}$ and $S_{i,j}$ denote the action space and infoset space in the $j$-th time window of player $i$. Finally the average strategy in the last window is outputted when the convergence condition is met. If XDO converges, the last-window average strategy is $\epsilon_0/2^k$-NE. To investigate the complexity of reaching $\epsilon$-NE, it is assumed without loss of generality that $\epsilon_0/2^k\leq \epsilon$.

Thus, RMDO can generalize to XDO with $m(j)= 4^j|S_{i,j}|^2|A_{i,j}|/\epsilon_0^2$ and the last-window average strategy. Based on Theorem \ref{theo:rmdo_sc}, we can determine the expected iterations and sample complexity for XDO to reach $\epsilon$-NE.

\begin{proposition}
\label{theo:xdo_cr}
XDO is an instance of RMDO. In the worst case, its frequency function $m(j)= 4^j|S_{i,j}|^2|A_{i,j}|/\epsilon_0^2$, where $j=0,1,\cdots, k-1$. Thus the sample complexity bound of XDO to reach $\epsilon$-NE is 
\begin{align}
    \tilde{\mathcal{O}}(k|A|X^3/\epsilon^2 + |A|X^3 4^k/\epsilon_0^2).
\end{align}
\end{proposition}
Proposition \ref{theo:xdo_cr} is a specific instance of Theorem \ref{theo:rmdo_sc} with an appropriate choice of $m(j)$. Theorem \ref{lemma:x} states that $k\leq |S|$; thus, theoretically, the restricted game stopping condition of XDO decays exponentially, leading to that in the worst-case scenario when $k=|S|$, XDO has an exponential sample complexity in the number of infosests. Therefore, XDO suffers from a large theoretical sample complexity. We demonstrate the algorithmic distinctions between existing DO and RMDO via a flowchart presented in Figure \ref{fig:rmdo}. Subsequently, we delve into an in-depth examination of the sources contributing to the sample complexity of existing RMDO instances, and explore potential approaches to mitigate and minimize these complexities in the following sections. We investigate from the perspectives of three elements in the sample complexity (equation \ref{eq:rmdo_sc}):




\begin{itemize}
    \item The selection of the frequency function $m(j)$ significantly impacts the performance of RMDO. In XDO, the exponential frequency function, caused by exponentially decaying stopping threshold for restricted games, contributes to an exponential sample complexity in terms of $|S|$ in the worst case. In Section \ref{sec:ins}, we will propose two instances of RMDO that only have polynomial sample complexity.
    \item The multiplier $k$ appearing in sample complexity is partially due to the cold starting of solving restricted games. Specifically, in each instance of solving a new restricted game, RMDO reset strategy and cumulative regret, resulting in $k$ independent procedures of game solving. Consequently, the appearance of $k$ in the sample complexity arises and thus cause potentially large complexity since $k$ is only bounded by $|S|$. In section \ref{sec:dows}, we present an algorithmic design called warm starting to address this. \looseness=-1
    \item The power of the dominating term $X$ (3 in XDO and ODO) is still high. Inspired by the sample efficient method, stochastic regret minimization, in Section \ref{sec:srmdo}, we propose Stochastic Regret-Minimizing Double Oracle as a solution to enhance scalability and reduce the the high power of $X$ in the complexity analysis. \looseness=-1
\end{itemize}

\subsection{Optimal Schemes of Frequency Function}
\label{sec:ins}

The exponentially growing frequency function $m(j)$ of XDO leads to an exponential increase in sample complexity with respect to $k$. On the other hand, XODO's inflexibility arises from the fact that it performs best response computation in each iteration, neglecting the balance between regret minimization and best response computation. In this section, to mitigate the large increase in sample complexity caused by a large value of $k$, and to balance the two computations, we present two instances of Regret-Minimizing Double Oracle (RMDO) designed to achieve $\epsilon$-Nash Equilibrium with only polynomial sample complexity. The first naive instance is to simply set a constant frequency of restricted game expanding, named by Periodic Double Oracle (PDO). Although PDO has reduced the sample complexity to polynomial, tuning such frequency constant is hard since in different game with different size, we need to retune it, making it hard to generalize. Thus, we finally propose Adaptive Double Oracle (AdaDO), featuring an \emph{optimal} expansion frequency function adaptively for restricted game with different size. Therefore, the tuning is less demanding. Additionally, the frequency function of AdaDO is provably the optimal choice in terms of theoretical sample complexity.

\subsubsection{Periodic Double Oracle}
Periodic Double Oracle (PDO) is derived from Regret-Minimizing Double Oracle (RMDO) by introducing a constant frequency value, denoted as $m(j) = c>1$. In practice, we treat $c$ as a hyperparameter that can be tuned. PDO reaches NE with its last-window average strategy, so we can derive the sample complexity using Theorem \ref{theo:rmdo_sc}.

\begin{proposition}
\label{theo:pdo_sc}
Given a constant hyperparameter $c$, since PDO computes BR every $c$ iterations, it is an instance of RMDO when $m(\cdot) \equiv c$ and its sample complexity to reach $\epsilon$-NE: 
\begin{align}
\tilde{\mathcal{O}}(k|A|X^3/\epsilon^2 + ckX + k|A|X^3/c\epsilon^2).
\end{align}
\end{proposition}

Periodic Double Oracle (PDO) is more sample efficient than previous the Double Oracle variants by simply adopting a constant frequency $m(j)=c>1$, and tuning this hyperparameter $c$ during execution to solve a game. In comparison to XODO, since $c>1$, PDO's upper bound for the dominant term in sample complexity is strictly less than that of XODO. Additionally, PDO exhibits a sample complexity only linear in $k$, making it significantly more sample-efficient than XDO, as it eliminates the exponential term in $k$, and slightly better than ODO, whose sample complexity has $k^2$ term.

However, tuning the hyperparameter $c$ in PDO can be challenging in practice. Its sensitivity to game size makes it difficult to choose an appropriate $c$, potentially leading to a large sample complexity. PDO, with the same periodicity, may not consistently perform well across different games, especially when their sizes vary (will be discussed in Section \ref{sec:exp}). Additionally, employing a fixed frequency for all restricted games with difference size may be overly simplistic and result in suboptimal performance. Therefore, it is preferable to select a frequency that adapts to the characteristics of current restricted game. In the next section, we introduce a new instantiation of RMDO that avoids the cumbersome hyperparameters tuning while maintaining a small sample complexity with such adaptive frequency function.


\subsubsection{Adaptive Double Oracle}
\label{sec:adado}
In this section, we propose a new instantiation of RMDO that has the lowest sample complexity among all RMDO applying different expansion frequencies, called \textbf{Adaptive Double Oracle (AdaDO)}. AdaDO has dynamic periodicity and reach NE with last-window average strategy. We first introduce the adaptive frequency function for AdaDO:
\begin{definition}[Adaptive Frequency Function]
Denote $|A_{j}|$ as $\max_{i\in \mathcal{P}, s \in S_i}|A_{i,j}(s)|$, where $A_{i,j}(s),\ S_{i,j}$ are defined as the set of actions at infostate $s$ and the set of infosets in time window $T_j$, respectively. The adaptive frequency function of AdaDO is
\begin{align}
\label{eq:adado_frequency}
    m(j) = \sqrt{|A_{j}|}\sum_i|S_{i,j}|/\epsilon.
\end{align}
\end{definition}
It's worth noting that the choice of $m(j)$ in this context is dependent on $\epsilon$, allowing us to set it according to the desired precision of the result. Furthermore, the frequency function $m(j)$ becomes window-dependent. In contrast to PDO, which selects a frequency function and maintains a fixed periodicity value across all windows, our approach involves computing the statistics of the restricted game in the current window and adjusting the frequency accordingly. \looseness=-1

The formal algorithm is outlined as follows. Initialize the strategy and restricted game using the standard Double Oracle method. Upon entering a new restricted game, perform one iteration of Counterfactual Regret Minimization (CFR). This step involves traversing the entire game tree of the current restricted game, allowing for the computation of $|A_{j}|$ (the action space size) and $\sum_i|S_{i,j}|$ (the sum of information set sizes).Compute the frequency of Best Response computation in the current window, denoted as $m(j)-1$, where subtracting one accounts for the fact that the first iteration of regret minimization is dedicated to obtaining game-related statistics $|A_{j}|$ and $\sum_i|S_{i,j}|$. Proceed with the remaining steps of the algorithm, consistent with other Regret-Minimizing Double Oracle (RMDO) methods.

Now we investigate AdaDO's sample complexity. We prove that such adpative frequency function is the optimal choice in terms of the sample complexity.
\begin{proposition}
\label{theo:adado_sc}
AdaDO is an instance of RMDO when $m(j)$ satisfy equation (\ref{eq:adado_frequency}). Thus the sample complexity of AdaDO matches the sample complexity \textbf{lower bound} of RMDO for all frequency function:
\begin{align}
    \tilde{\mathcal{O}}(k|A| X^3/\epsilon^2 + 2k\sqrt{|A|}X^2/\epsilon).
\end{align}
\end{proposition}

In theory, among all Double Oracle methods that instantiated from RMDO with different frequency function, AdaDO is provably the most sample efficient method. Compared to PDO, AdaDO removes part of the dominating term in the sample complexity of PDO, $\mathcal{O}(k|A|X^3/c\epsilon^2)$. Although it has not eliminated all the cubic terms of $X$, the complexity has been significantly reduced by merely picking a frequency function without changing regret minimizers or best responders. 

A potential empirical problem when applying AdaDO is that the choice of adaptive frequency function in AdaDO is derived by reducing a theoretical complexity. But the empirical complexity can be much smaller, making AdaDO which has the theoretically optimal frequency scheme perform suboptimal empirically. This phenomenon aligns with previous discussions on the impact of this gap on algorithmic design for extensive-form games, as explored in prior works like the paper of ~\citet{brown2016strategy}. To address this, in the practical version shown in Algorithm \ref{alg:adado}, we adopt empirical frequency function, a discounted frequency function $\hat{m}(\cdot) = \alpha \cdot m(\cdot)$, where $\alpha \in (0,1]$. Besides, we execute early stop of restricted game regret minimization when the exploitability is decreasing slowly.

\begin{algorithm}[t!]
\begin{algorithmic}
    \caption{Adaptive Double Oracle (Practical verion)}
    \label{alg:adado}
    \STATE \textbf{Input:} Empirical frequency function of AdaDO $\hat{m}(\cdot)=\alpha \cdot m(\cdot)$, $\alpha\in (0, 1]$, $m(\cdot)$ in equation (\ref{eq:adado_frequency}), early stop tolerance $\delta$, exploitability checking frequency $c$ for early stop.
    \STATE $\Pi_1=\mathbb{BR}_i(\pi^0)$ for $i \in \{1,2\}$.
    \STATE Construct restricted game $\mathbf{G}_1$ with $\Pi_1$.
    \FOR{$t=1,\cdots,\infty$}
    \STATE Run one iteration of CFR in $\mathbf{G}_t$.
    \IF{$t \mod \hat{m}(j)=0$ or ($t \mod c=0$ and $|e(\Tilde{\pi}^{t}) - e(\Tilde{\pi}^{t})| < \delta$)}
    \STATE Compute average strategy $\Tilde{\pi}_i^t=\sum_{t\in T_j}\pi^t/|T_j|$.
    \STATE
    $\Pi_{t+1} = \Pi_t \cup  \mathbf{BR}_i(\Tilde{\pi}^t_{-i})$ for $i \in \{1,2\}$.
    \IF{$\Pi_{t+1} \neq \Pi_{t}$}
    \STATE Start new window: $j = j+1$.
    \STATE Reset strategy $\pi^{t+1}$.
    \STATE Construct restricted game $\mathbf{G}_{t+1}$ with $\Pi_{t+1}$.
    \ENDIF
    \ENDIF
    \ENDFOR
\end{algorithmic}
\end{algorithm}

AdaDO is not the first RMDO method employing dynamic frequency function. XDO, which largely follows the original DO process, passively pick a dynamic frequency function. Specifically, in each restricted game, XDO refrains from computing the Best Response and continues regret minimization until its average strategy reaches local $\epsilon$-NE of that restricted game. As analyzed in Section \ref{sec:rmdo}, the threshold $\epsilon$ in XDO decreases exponentially with the time window, leading to its exponential sample complexity. In contrast, AdaDO actively select the frequency function to reach the lower bound of the sample complexity of RMDO, which is only linear in $k$ and polynomial in $|S|$ in the worst case.

\subsection{Comparison to Regret Minimization}
\label{sec:compare_rm}
We then proceed to compare Counterfactual Regret Minimization (CFR) with the newly proposed instances of Regret-Minimizing Double Oracle (RMDO). While theoretical complexities are not necessarily indicative of reaching approximate Nash Equilibrium (NE), we still consider sample complexity as a comparable metric for CFR and existing DO methods, given that we compute complexities in a similar manner.

In general, Double Oracle methods are known to be efficient in games with NE characterized by small support. To illustrate this, we compare the dominant term of CFR's complexity $\tilde{\mathcal{O}}(|S|^3 |A|/\epsilon^2)$ with PDO/AdaDO's complexity $\tilde{\mathcal{O}}(k|A|X^3/\epsilon^2)$, where $X$ is the largest size of the games constructed by the support of NEs. When the NE supports are small enough to satisfy $\tilde{\mathcal{O}}(|S|^3) > \tilde{\mathcal{O}}(kX^{3})$, and thus $\tilde{\mathcal{O}}(|S|^3|A|/\epsilon^2) > \tilde{\mathcal{O}}(kX^3|A|/\epsilon^2)$, the dominating term of sample complexities of PDO and AdaDO are less than that of CFR under this small NE support condition. Thus, this comparison confirms that Double Oracle methods exhibit lower sample complexity when the NE support is small. This represents, to our knowledge, the first discussion in DO literature regarding the conditions under which DO outperforms regret minimization methods theoretically.

\section{Double Oracle with Warm Starting and Stochastic Regret Minimizer}
In this section, we propose two improvements for all Double Oracle methods to reduce the complexity caused by $k$ and $X$ in the sample complexity of RMDO (equation \ref{eq:rmdo_sc}), as discussed at the end of section \ref{sec:existing_do}.
\subsection{Warm Starting}
\label{sec:dows}
\begin{figure}[t!]
    \centering
    \includegraphics[width=\textwidth]{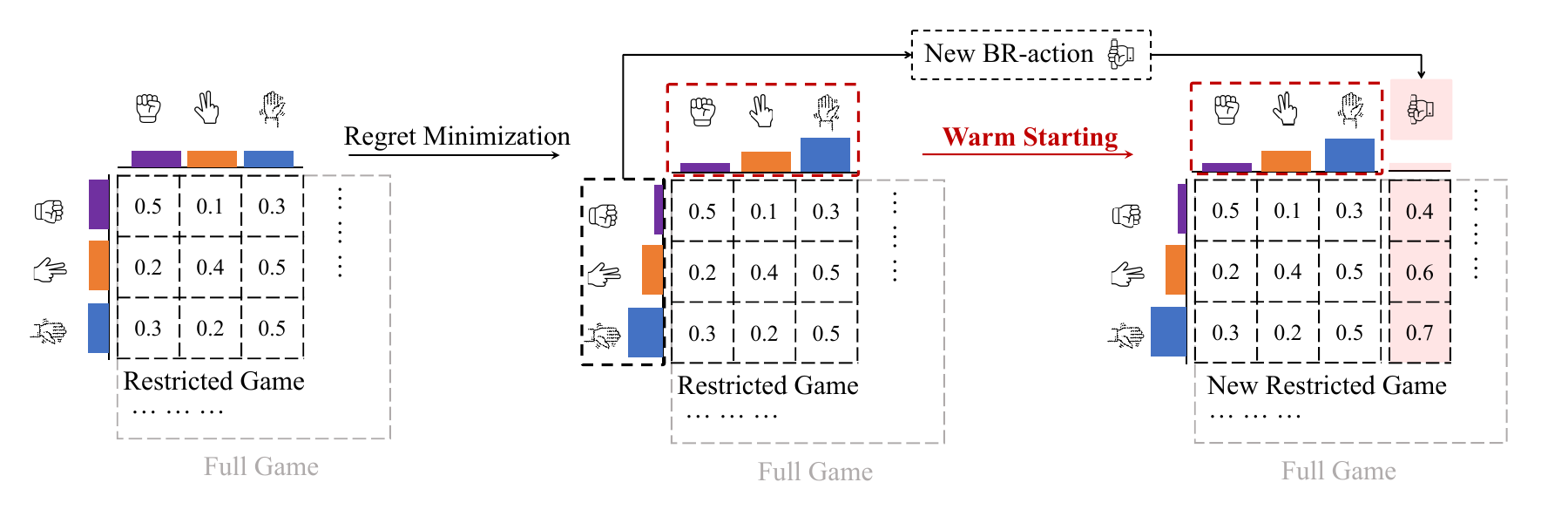}
    \caption{Restricted game expanding and warm starting of Regret Minimizing Double Oracle in EFGs. In restricted game, the regret minimizer will keep updating the regret and average strategy. After $m(\cdot)$ iterations of regret minimization, we compute best response actions (BR) against the restricted strategy (violet, orange and blue bars). If there are new BR actions, we expand the restricted game with them.}
    \label{fig:enter-label}
\end{figure}

Given that RMDO initiates training from scratch for each of the $k$ restricted games without transferring knowledge from the previous ones, this section introduces an approach to enhance convergence speed. We propose to integrate the Double Oracle framework with \textbf{Warm Starting}, a technique that involves transferring regret learned in prior restricted game to the later one for faster convergence. The empirical performance of this approach will be demonstrated in Section \ref{sec:experiments_ws}, where we will show substantial improvements in most games. \looseness=-1

We first revisit the details of restricted game expansion and current initialization methods in the Double Oracle framework. In Figure \ref{fig:enter-label}, nodes symbolize information states, and edges represent actions. Regret-Minimizing Double Oracle partitions the set of training iterations, $\{t|t=1,2,\cdots,T\}$ with time window, denoted as $T_j$, $j=1,2\cdots, k$. At iterations $t$, we represent the restricted game as $\mathbf{G}_t$, constructed by the grey lines in the figure. For a given $j$, we have $\mathbf{G}_t=\mathbf{G}_{t'}$ for all $t\in T_j$. The dotted lines signify actions not included in the restricted games, indicating that we will not traverse those branches when solving the restricted games. Consequently, the regret and average strategy undergo continuous updates within $T_j$ through regret minimizers. This update process concludes when computing best response actions against $\bar{\pi}$ reveals that $\mathbb{BR}(\bar{\pi}) \not\in \mathbf{G}_t$. RMDO employs $m(j)$, a mapping from the index $j$ of the current time window $T_j$ to the frequency of computing the best response within that window. This frequency denotes how often we perform the check. Suppose, at iteration $t$, we compute the best response and discover a new action not present in the current restricted game. In such a scenario, we add these new best response actions (depicted as orange edges) to $\mathbf{G}_t$, resulting in a new restricted game $\mathbf{G}_{t+1}$. Subsequently, a new time window $T_{j+1}$ is initiated.

When initiating a new restricted game, the standard procedure involves clearing the regret and initializing the average strategy with a uniform distribution over actions. This inhibits the utilization of prior knowledge acquired during the learning of previous restricted games. Consequently, it leads to the learning of $k$ independent games, resulting in the appearance or a high order of $k$ in sample complexity. Since the regret in the previous restricted games indicates the values of different branches. Clearing regrets upon entering a new restricted game necessitates relearning the values of old branches, which is inefficient and waste of computational resources. This inefficiency can be mitigated through warm starting, where RMDO no longer treats the solving of $k$ independent restricted games as separate tasks, but rather leverages prior knowledge for faster convergence. 

To enhance the exploitation of knowledge learned in previous restricted games, we propose the incorporation of warm starting for regret when entering a new restricted game. Specifically, we initialize the counterfactual regret of the actions which appeared in the previous restricted game with their previous regrets. Such initialization is reasonable since the restricted games have the relation of $\mathbf{G}_1 \subset \mathbf{G}_2 \subset \cdots \subset \mathbf{G}_k$. Additionally, We warm start the cumulative strategy with that of the previous restricted game. For new actions emerging in the current restricted game, we initialize their regret and strategy with fixed values $\varepsilon>0$. Suppose at iteration $t$ there is new added BR actions and new constructed restricted game $\mathbf{G}_j$. Thus we initialize the regret of $\forall s,a \in \mathbf{G}_j$ with \looseness=-1
\begin{align}
{R}_i^{t}(s,a) = \Bigg\{
    \begin{array}{cc}
         {R}_i^{t-1}(s,a) , & s,a \in \mathbf{G}_{j-1}  \\
         \varepsilon ,& \text{otherwise}
    \end{array}.
\end{align}
Although warm starting is designed intentionally to reduce the complexity caused by $k$, the theoretical improvements are hard to prove due to the inconsistent action and infoset space across different restricted games, and the way we initialize regret and compute average strategy. However, in section \ref{sec:exp}, we observe empirically in many research games that warm starting help faster convergence significantly.

Replacing cold starting with warm starting can implicitly mitigate the sample complexity caused by $k$. Another contribution to the sample complexity of RMDO is $X$. It arises from the need to traverse the entire restricted game tree during regret minimization. This becomes particularly pronounced in large games, where even the restricted game can become substantial in size. To tackle this issue, we propose a scalable RMDO framework, Stochastic Regret-Minimizing Double Oracle. This framework leverages stochastic regret minimization techniques, specifically Monte-Carlo Counterfactual Regret Minimization (MCCFR), and incorporates approximate best response methods. The objective of this framework is to address the scalability concerns associated with RMDO and enable more efficient handling of large-scale games.

\subsection{Stochastic Regret-Minimizing Double Oracle}
\label{sec:srmdo}

The core idea behind \textbf{Stochastic Regret-Minimizing Double Oracle (SRMDO)} is to replace regret minimization in the RMDO framework with stochastic regret minimization (SRM) for the restricted game solving. While computing the oracle Best Response (BR) still requires traversing the entire game tree, the number of times computing BR is significantly less than that of executing regret minimization. We also introduce the sample complexity associated with leveraging approximate BR in SRMDO to enhance its scalability. Importantly, we demonstrate that replacing the regret minimizer with Monte-Carlo Counterfactual Regret Minimization (MCCFR) can help AdaDO reduce the theoretical sample complexity by $\mathcal{O}(X)$.

In this paper, we employ outcome-sampling Monte-Carlo Counterfactual Regret Minimization (MCCFR) ~\citep{mccfr} as the stochastic regret minimizer. This choice is made to ensure maximum scalability, as outcome-sampling MCCFR is the only method in the CFR family that learns from bandit feedback. Regarding the approximate best responder, any method, including Reinforcement Learning or No-Regret Learning, can be applied. However, it must be capable of reaching high precision $\delta$-Best Response (with high probability), where $\delta < \epsilon$. The complexity of computing one iteration of approximate Best Response computation is denoted as $\mathcal{O}(H)$.
\begin{theorem} 
The Last-window average strategy of SRMDO has the following sample complexity to reach $\epsilon$-Nash Equilibrium when employing oracle best responses:
\begin{align}
\label{eq:srmdo_sc}
    \tilde{\mathcal{O}}( kH|A|X^2/\epsilon^2 + \sum_{j=1}^k H_{j} m(j) +   |A|X^3/\epsilon^2m(j)),
\end{align}
and the following sample complexity when employing $\delta$-BR approximate best responder:
\begin{align}
    \tilde{\mathcal{O}}( kH|A|X^2/(\epsilon - \delta)^2 + \sum_{j=1}^k H_{j} m(j) +   H|A|X^2/(\epsilon - \delta)^2m(j))
\end{align}
, where $\delta < \epsilon$ is required, $H_{j} = \max_i H_{i,j}$ is the largest horizon of the restricted game in $T_j$, and obviously $H_{j} \ll |S_{i,j}|$.
\label{theo:SRMDO_ebr}
\end{theorem}

The sample complexity of Stochastic Regret-Minimizing Double Oracle (SRMDO) has been significantly reduced by decreasing the power of $X$ in the term that is independent of the frequency function from $3$ to $2$. This reduction is crucial, as even with the optimal choice of frequency function, the power of the first term cannot be influenced. Additionally, it is observed that employing approximate Best Response can further reduce the remaining term from $\tilde{O}(X^3)$ to $\tilde{O}(X^2)$. However, the benefit introduced by approximate Best Response is relatively small, as selecting an appropriate frequency function can easily achieve the same power reduction of the later two terms in equation (\ref{eq:srmdo_sc}). Therefore, for the remainder of this paper, we will use oracle Best Response for simplicity.

It is also reasonable to apply periodic and adaptive frequency functions in SRMDO for $m(j)$, which we refer to as \textbf{Stochastic Periodic Double Oracle (SPDO)} and \textbf{Stochastic Adaptive Double Oracle (SADO)}, respectively. The extension of PDO to SPDO and won't be discussed in detail here. To get the sample complexity of SPDO, we only need to set $m(j)=c$ in equation \ref{eq:srmdo_sc}. We demonstrate SADO here since it requires a new adaptive function that can help decrease the power of $X$ in all the dominating terms from $3$ to less or equal to $2$.

\begin{theorem}[SADO]
By employing oracle best responses and the following frequency function
\begin{align}
    m(j) = \sqrt{\frac{|A_j|(\sum_i|S_{i,j}|)^3}{H_{j}\epsilon^2}},
\end{align}
the sample complexity of Last-window average strategy in SRMDO is
\begin{align}
\tilde{\mathcal{O}}( k|A|HX^2/\epsilon^2 + 2k\sqrt{|A|H}X^{1.5}/\epsilon).
\end{align}
\label{theo:srmdo_ada}
\end{theorem}
The theorem above indicate that despite the necessity for the full traversal of the game tree with oracle Best Response, the complexity of Stochastic Regret-Minimizing Double Oracle (SRMDO) with an appropriate adaptive frequency function can still be reduced by $\mathcal{O}(X)$.

\section{Experiments}
\label{sec:exp}
In this section, we first examine the efficiency of Periodic Double Oracle (PDO) compared to other baselines including XDO, XODO, and regret minimization methods for game solving. We also analyze the support of PDO in different games. Subsequently, we explore the effect of warm starting and Adaptive Double Oracle (AdaDO). Finally, we delve into the empirical performance of Stochastic Regret-Minimizing Double Oracle, encompassing Stochastic Periodic Double Oracle (SPDO) and Stochastic Adaptive Double Oracle (SADO).

To assess the efficiency in game solving, we evaluate performance through plots of exploitability (the distance to Nash Equilibrium) versus the number of visited nodes. Such plots offer a clear visualization of the required complexity to achieve specific precision of Nash Equilibrium. Here we only include the visited nodes in the algorithms including both (stochastic) regret minimization and best response computation, but exclude the complexity in computing the exploitability for fair comparison. Additionally, in the stochastic regret minimization setting, we provide the number of visited nodes required to reach a specific level of exploitability for different algorithms.

We conducted tests on perfect and imperfect-information extensive-form poker games, including Blotto, Kuhn Poker, Leduc Poker, and their variants, namely Large Kuhn Poker, Leduc Poker Dummy, and Leduc Poker with $10$ cards. Detailed descriptions of the games can be found in Appendix \ref{append:games}. The implementation is primarily based on the library OpenSpiel ~\citep{lanctot2019openspiel}. The selected baselines include Extensive-Form Fictitious Self-Player (XFP) ~\citep{nfsp} and Linear Counterfactual Regret Minimization (LCFR) ~\cite{discfr} for regret minimization setting, and Outcome-Sampling Monte-Carlo Counterfactual Regret Minimization (simplified as MCCFR) for the stochastic regret minimization setting.

\subsection{Periodic Double Oracle}
\begin{figure}[t!]
     \centering
    \includegraphics[width=.99\textwidth]{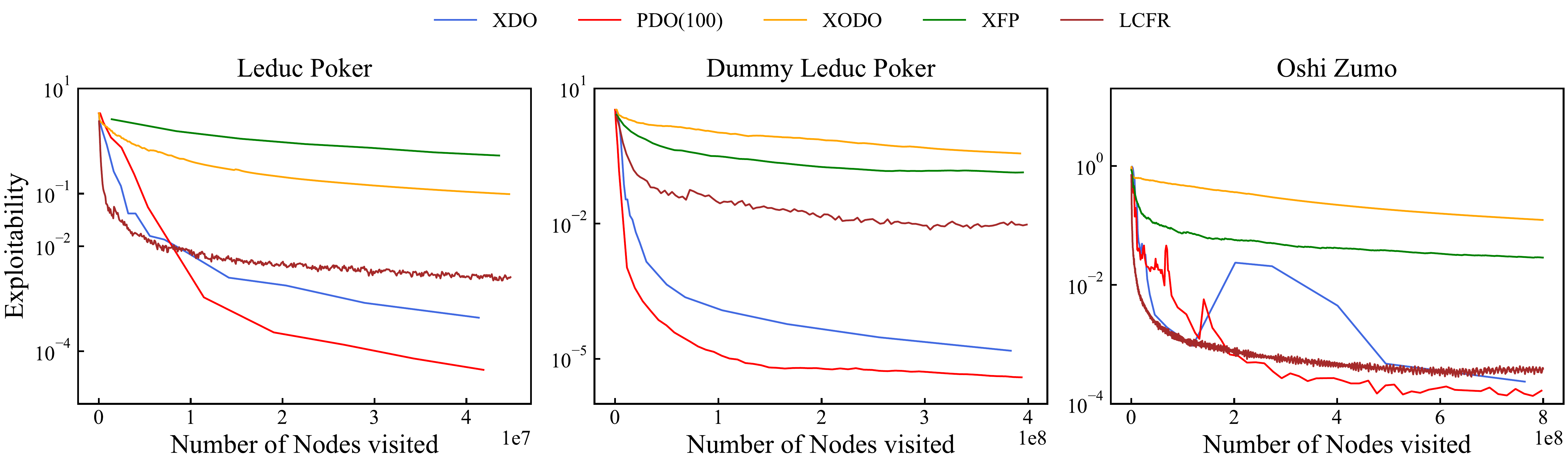}
\caption{Exploitability-Visited Nodes Performance of Extensive-form Double Oracle (XDO), Periodic Double Oracle with its periodicity, i.e. PDO($c$), Extensive-form Online Double Oracle (XODO), Extensive-form Fictitious Self-Play (XFP), and Linear Counterfactual Regret Minimization (LCFR). Our algorithm PDO achieves the lower exploitability than any other methods.}
  \label{fig:exploitability_all}
\end{figure}

We initially examine the performance of well-tuned Periodic Double Oracle (PDO) in comparison to baselines and existing Regret-Minimizing Double Oracle (RMDO) instances (Figure \ref{fig:exploitability_all}). Our findings reveal that PDO surpasses XDO, XODO, and baseline methods (XFP and LCFR) by a significant margin in games such as Leduc Poker, Leduc Poker Dummy, and Oshi Zumo. PDO has a more stable exploitability curve and a faster convergence in general compared to XDO. Specifically, PDO achieves $10^{-4}$ exploitability faster than any other methods, even though it may not perform as well as LCFR in the early stages of Leduc Poker and Oshi Zumo. In summary, PDO is capable of converging faster to lower exploitability.


\begin{table}[t!]
\centering
\begin{tabular}{@{}lllllllll@{}}
\toprule
Games                   & \multicolumn{2}{l}{\emph{Large Kuhn Poker}} &  \multicolumn{2}{l}{\emph{Leduc Poker}} & \multicolumn{2}{l}{\emph{Kuhn Poker}} & \multicolumn{2}{l}{\emph{Oshi Zumo}} \\ \midrule
                        & PDO                   & LCFR                       & PDO                          & LCFR      & PDO                          & LCFR     & PDO                         & LCFR     \\
Min. Support & 50\%                  & 100\%                     & \multicolumn{1}{c}{33\%}     & 100\%     & \multicolumn{1}{c}{50\%}     & 100\%    & \multicolumn{1}{c}{25\%}    & 100\%    \\
Avg. Support & 76\%                  & 100\%                  & \multicolumn{1}{c}{85\%}     & 100\%     & \multicolumn{1}{c}{83\%}     & 100\%    & \multicolumn{1}{c}{93\%}    & 100\%    \\ \bottomrule
\end{tabular}
\caption{Table of minimum and average support percentages of well-tuned Periodic Double Oracle (PDO) when first reaching $10^{-3}$-NE. Specifically, the minimum support percentage is defined as $\min_{s \in S} \text{supp}^{\bar{\pi}}(s) / |A(s)|$, and Average Support Percentage $\sum_{s \in S}\text{supp}^{\bar{\pi}}(s) / |A(s)||S|$. It is worth noting that even though the average support percentage is close to 1 in some games, the tree structure of Extensive-Form Games (EFGs) allows for efficient traversal, where only one action with zero probability can reduce the complexity of visiting the entire subtree rooted by the outcome state of this action. }
\label{table:support_full}
\end{table}

We then study the support of Double Oracle (DO) methods by investigating the support of average strategies of the well-tuned PDO when reaching $\epsilon$-NE, where $\epsilon=10^{-3}$. The metrics include the minimum support percentage, defined as $\min_{s \in S} \text{supp}^{\bar{\pi}}(s) / |A(s)|$, and Average Support Percentage $\sum_{s \in S}\text{supp}^{\bar{\pi}}(s) / |A(s)||S|$.

The analysis of the minimum support in Double Oracle (DO) strategies reveals a noteworthy efficiency compared to Linear CFR. The average strategy of Linear CFR consistently maintains full support, implying that it assigns non-zero probability to every action. In contrast, DO achieves significantly lower minimum support percentages. This efficiency is indicative of the learning process, as a lower support implies the need to learn the distribution over fewer actions, facilitating faster convergence. Additionally and notably, this table also implies that DO tends to produce more sparse solution (i.e. with less actions with positive probability). 

While the mean support may not exhibit a substantial difference from that of Linear CFR, it's crucial to note that this metric operates at the infostate level. In the broader context of the entire game tree of an Extensive-Form Game (EFG), actions excluded from the strategy contribute to pruning entire subtrees rooted by them. This efficient pruning mechanism is expected to enhance sample complexity significantly, showcasing the reason to the efficiency DO in games with small support NE.

\subsection{Adaptive Double Oracle and Warm Starting Double Oracle}
\label{sec:experiments_ws}
The evaluation of Adaptive Double Oracle (AdaDO) and warm starting (simplified as suffix -WS) in Figure \ref{fig:exploitability_ws_ada_all} provides insights into their efficiency compared to PDO and Linear CFR in different poker games. For PDO, we directly use a well-tuned periodicity $c=100$.

We first introduce the performance of AdaDO in Figure \ref{fig:exploitability_ws_ada_all}. In Blotto, Leduc Poker and Leduc Poker Dummy, and large Kuhn Poker, the exploitability of AdaDO decreases faster than that of PDO and LCFR. In other games, the exploitability of AdaDO converges to the same magnitude of exploitability to PDO. In summary, AdaDO performs better than PDO in most research games and significantly outperforms LCFR in all experiments. It is impressive to observe that AdaDO, not only avoid extensive tuning on the hyperparameter periodicity as PDO, but can perform faster convergence than PDO and other baselines as well. Since in Figure \ref{fig:exploitability_all}, PDO outperforms other baselines already, AdaDO is now the more efficient method that has the state-of-the-art performance.\looseness=-1

We then investigate the effectiveness of warm starting in Figure \ref{fig:exploitability_ws_ada_all}. In Blotto, Leduc Poker $10$ cards and Large Kuhn Poker, warm starting significantly reduces the exploitability of Double Oracle methods. Especially in Large Kuhn Poker, both PDO and AdaDO exhibit early drops in exploitability, reaching values over $10^8$ less than those of Linear CFR. This suggests that warm starting can accelerate the convergence of Double Oracle. In other games, warm starting has little impact on DO methods. In summary, warm starting overall has positive influence, and in specific games can help reduce the exploitability significantly faster and at most $8$ levels of magnitude less exploitable, i.e. higher precision. \looseness=-1

Overall, AdaDO, leveraging a theoretically optimal frequency function, performs faster convergence than a well-tuned PDO in most games. The impact of warm starting is positive in most games, being more pronounced in Large Kuhn Poker. Importantly, AdaDO achieves these results without hyperparameter tuning, highlighting its potential for efficient use. \looseness=-1

\begin{figure}[t!]
     \centering
    \includegraphics[width=.99\textwidth]{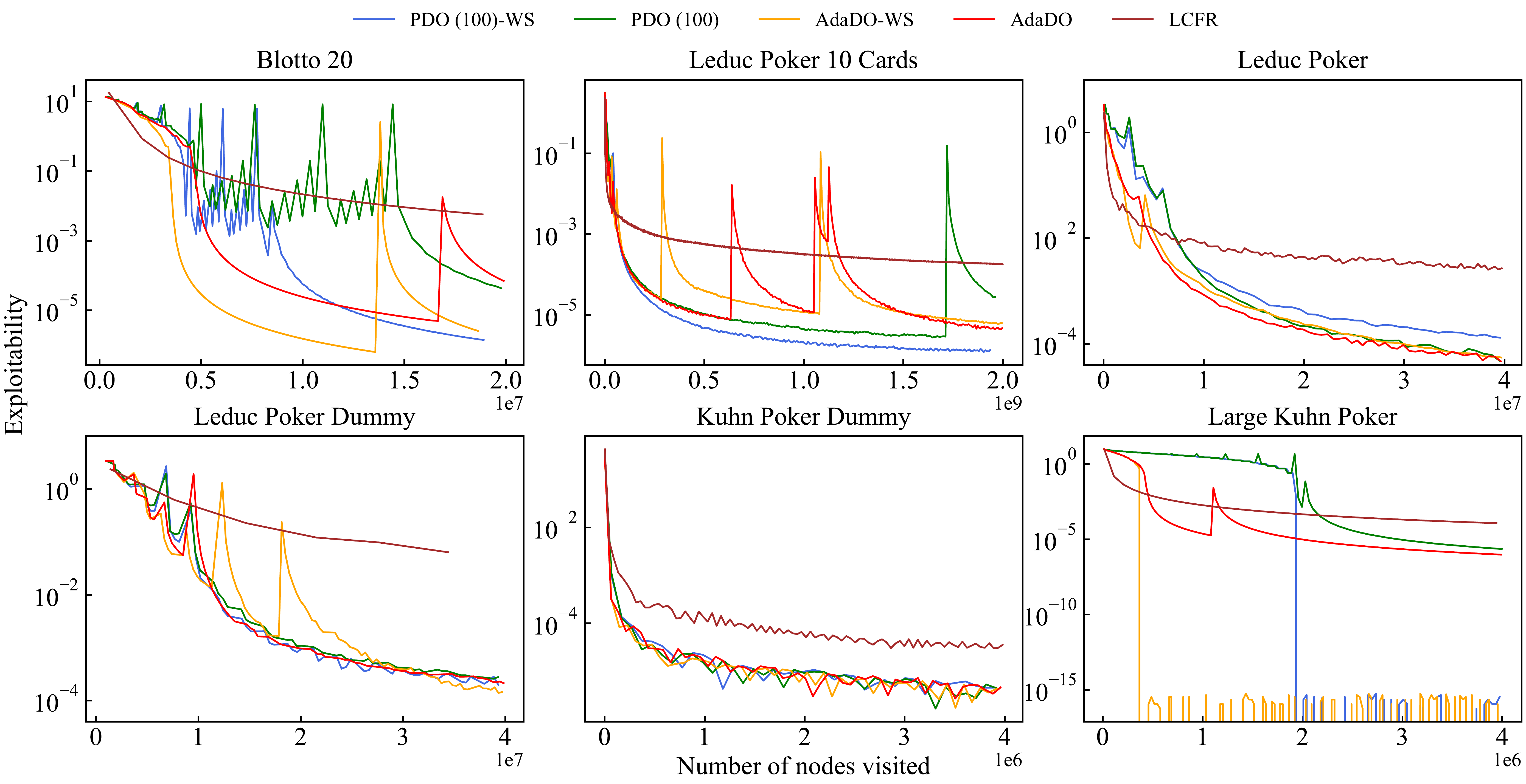}
\caption{Exploitability-Visited Nodes Performance of PDO with and without warm starting, AdaDO with and without warm starting, and LCFR. Warm starting help reduce exploitability significantly in Blotto and Large Kuhn Poker. AdaDO outperforms LCFR and PDO in most games.}
  \label{fig:exploitability_ws_ada_all}
\end{figure}

\subsection{Stochastic Regret-Minimizing Double Oracle}
In this section, We display empirical assessments of the Stochastic Regret-Minimizing Double Oracle (SRMDO), which leverage Outcome-Sampling MCCFR for restricted game solving. Our baseline is also Outcome-Sampling MCCFR, thus the comparison between them is fair. We employ oracle Best Response for all instances due to simplicity and also the fact that sample complexity won't change significantly by adopting approximate BR. 

\begin{figure}[t!]
    \centering
    \includegraphics[width=\textwidth]{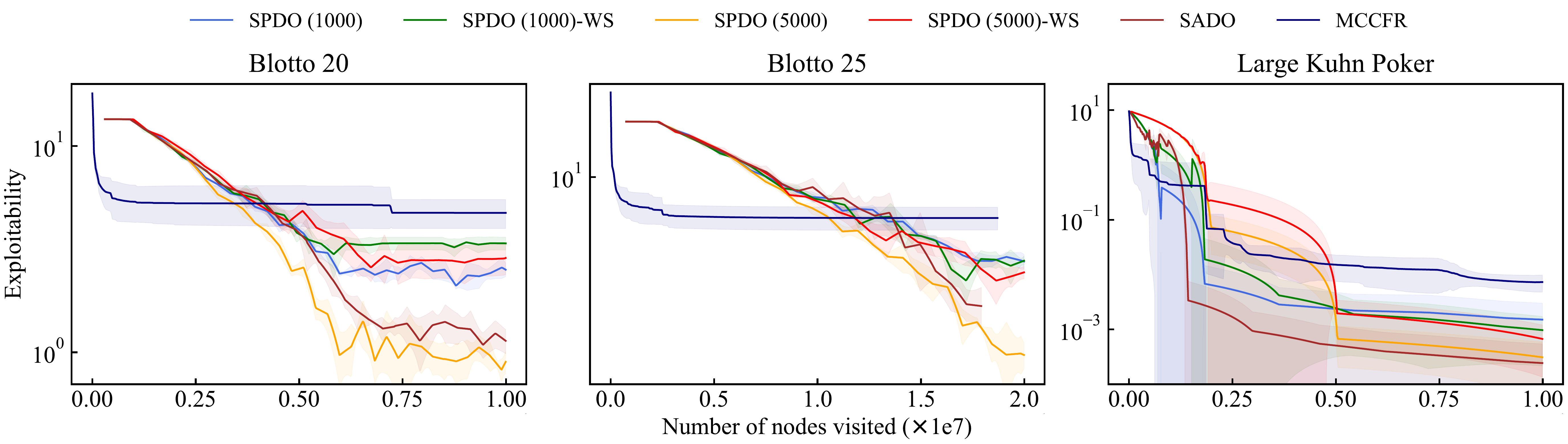}
    \caption{Exploitability experiments of Stochastic PDO (SPDO), and Stochastic Adaptive DO (SADO) with and without warm starting and Outcome-Sampling Monte-Carlo CFR (MCCFR). SPDO and SADO performs similarly good, and outperform MCCFR significantly.}
    \label{fig:srmdo}
\end{figure}

We study the performance of methods of SRMDO including Stochastic Periodic Double Oracle (SPDO) and Stochastic Adaptive Double Oracle (SADO) in Figure \ref{fig:srmdo}. In Blotto games, MCCFR experiences fast exploitability reduction in the early stages but gets stuck at low precision Nash Equilibrium. On the other hand, SPDO and SADO shows a significant improvement in exploitability at later stages, bringing its strategy closer to Nash Equilibrium compared to MCCFR. Specifically, in all experiments, SPDO with periodicity $5000$ and SADO converge to less exploitability than other algorithms. In Large Kuhn Poker, SADO outperforms SPDO, and surpass MCCFR by converging to $10$ times less exploitability solution. According to the results for the best SPDO and SADO to reach specific exploitability in Table \ref{table:srmdo}, SADO only needs less than half of visited nodes of SPDO and MCCFR to reach the same level of exploitability. According to the table, in Large Kuhn Poker and vanilla Kuhn Poker, SADO reaches exploitability $0.005$ and $0.0003$, respectively, with less than half number of visited nodes of SPDO (Table \ref{table:srmdo}). These results suggest the efficiency of Double Oracle and the potential of DO methods to scale. 

Though SADO performs similarly to SPDO in Blotto games and only outperform SPDO in Kuhn Pokers, the inefficiency of tuning periodicity hyperparameter in PDO is obvious. In Figure \ref{table:srmdo}, SPDO is sensitive to this periodicity since SPDO($1000$) and SPDO($5000$) has clearly distinct performance. This observation is crucial for understanding the challenges associated with tuning frequency hyperparameter in PDO and SPDO. The results highlight that the optimal choice of frequency can vary significantly among different games. So retuning $c$ is necessary when applying SPDO to new games, emphasizing the demanding nature of hyperparameter tuning in PDO. On the contrary, Stochastic Adaptive Double Oracle (SADO) is executed without extensive tuning but still demonstrates superior results.


\begin{table}[t!]
\centering

\begin{minipage}{0.5\textwidth}
\centering
\begin{tabular}{ccc}
\hline
                               & \multicolumn{2}{c}{\textit{Large Kuhn Poker}}                                                                                                                  \\
Exploitability                           & $0.5$                                                                         & $0.005$                                                                        \\ \hline
SPDO                           & $0.69$ \footnotesize{$(0.02)$}                               & $3.62 $ \footnotesize{($ 1.61$)}                              \\
\textbf{SADO} & \boldmath{$0.65 $} \footnotesize{($ 0.08$)} & \boldmath{$1.68 $} \footnotesize{($ 0.651$)} \\
MCCFR                          & \multicolumn{1}{l}{$0.69 $ ($ 0.62$)}                                         & $43.05$ \footnotesize{($ 44.88$)}                             \\ \hline
\end{tabular}
\end{minipage}%
\begin{minipage}{0.5\textwidth}
\centering
\begin{tabular}{ccc}
\hline
                               & \multicolumn{2}{c}{\textit{Kuhn Poker}}                                                                                                                       \\
Exploitability                           & $0.2$                                                                         & $0.0003$                                                                      \\ \hline
SPDO                           & $0.02 $ \footnotesize{($ 0.0 $)}                             & $95.48$ \footnotesize{($35.99$)}                             \\
\textbf{SADO} & \boldmath{$0.003$} \footnotesize{($0.001$)} & \boldmath{$34.59$} \footnotesize{($13.23$)} \\
MCCFR                          & $0.043 $ \footnotesize{($ 0.0$)}                             & $87.79$ \footnotesize{($75.33$)}                             \\ \hline
\end{tabular}
\end{minipage}%

\caption{Representative results of SRMDO: Number of visited nodes ($\times$1e6) of different algorithms to reach different exploitability, 5e-1 and 5e-3 in Large Kuhn Poker, and 2e-1 and 3e-4 in Kuhn Poker. SRMDO instances requires significantly less number of visited nodes to reach the same level of exploitability. \looseness=-1}
\label{table:srmdo}
\end{table}

Overall, the results highlight the effectiveness of SPDO and SADO. Additionally, the importance of hyperparameter tuning in PDO and SPDO emphasize the need for adaptive approaches like SADO to dispense the complexity of hyperparameter tuning. The potential advantages of SADO, including not requiring tuning and theoretical sample efficiency, make it a promising direction for addressing the challenges posed by extensive-form games.

\section{Conclusion}
In this paper we propose Regret-Minimizaing Double Oracle (RMDO) to unify the existing Double Oracle algorithms combined with regret minimization for theoretical study and further algorithmic improvement. Based on RMDO framework, we derive the sample complexity of existing methods Online Double Oracle and Extensive-Form Double Oracle, and prove that the later method, which is the state-of-the-art method for EFGs, suffers from exponential sample complexity. To address this problem, we propose two instances that only has polynomial sample complexity---Periodic Double Oracle, which requires hyperparameter tuning and Adaptive Double Oracle. These two innovative DO methods are provably more sample-efficient than the previous DO methods. To further make RMDO more sample efficient and scalable, we propose to adopt warm starting and stochastic regret minimizer for restricted game solving. In the empirical assessments, PDO can converge to lower exploitability than regret minimization methods and previous DO methods. While AdaDO which is less sensitive to the hyperparameter outperforms a well-tuned PDO in most environments. Moreover, AdaDO combined with warm starting and stochastic regret minimizer exhibits accelerated convergence in exploitability and reach the state-of-the-art performance in most games.

\newpage

\appendix
\section{Notations}

We offer the table of notations for easier understanding the proofs in the next sections.

\begin{table}[h]
\caption{Table of notations.}

\centering
\begin{tabular}{@{}ll@{}}
\toprule
Notations           & Descriptions                                                                      \\ \midrule
$s$                 & Information State/Information Set                                                 \\
$a$                 & Action                                                                            \\
$i$                 & Player                                                                            \\
$j$                 & Index of time windows/restricted games                                            \\
$t$                 & Iteration/time                                                                    \\
$\pi$               & Strategy/Policy                                                                   \\
$e(\pi)$            & Exploitability of $\pi$                                                           \\
$R_i^t(s, a)$       & Cumulative regret of player $i$ at iteration $t$ at infoset $s$ taking action $a$ \\
$\bar{\pi}^t(s, a)$ & Average strategy                                                                  \\
$A(s)$              & The set of action at infoset $s$                                                  \\
$S_i$               & The set of infosets of player $i$ in the original game                            \\
$A_{j}(s)$          & The set of action at infoset $s$ in time window $j$                               \\
$S_{i,j}$           & The set of infosets of player $i$ in time window $j$                              \\
$X$                 & $\max_{s \in S, \pi^* \in \Pi^*} \sum_s \text{supp}^{\pi^*}(s)$                   \\
$H_{i,j}$           & The largest length of all trajectories of player $i$ in $T_j$ in Stochastic RM    \\
$T$                 & Current time or current iterations of training                                    \\
$T_j$               & The $j$-th time window                                                            \\
$\mathbf{G}_j$      & The $j$-th restricted game, the restricted game in $T_j$                          \\
$k$                 & Number of time windows (restricted games)                                                              \\
$m(\cdot)$          & Frequency function                                                                \\ \bottomrule
\end{tabular}
\end{table}

\section{Proofs}
\label{app:proof}

\begin{theorem}
\label{theo:RMDO_oas_bound} 
In RMDO, suppose the regret minimizer has $\tilde{\mathcal{O}}(|S_i|\sqrt{|A|T})$ regret, the weighted-average regret bound of RMDO:
\begin{align}
\tilde{\mathcal{O}}(\sum_{j=0}^{k-2} \frac{|T_j|}{T}\cdot[m(j)-1] + \sum_{j=0}^{k-1}\frac{\sqrt{k}|S_{i,j}||T_j|}{T\sqrt{|T_j| - m(j) + 1}}),
\end{align}
converges to $0$ if $m(j)$ is sublinear in $T$.
\end{theorem}
\begin{proof}
    Please refer to the proof of Theorem 3.3 in the previous work ~\citep{tang2023regret}. The only difference is we replace $S_i$ with $S_{i,j}$, which can be easily shown reasonable.
\end{proof}

\begin{lemma}[Required Iterations]
\label{lemma:ri}
    It requires in the worst case the following number of iterations for LAS of RMDO to reach $\epsilon$-NE:
    \begin{align}
     \sum_{j=1}^k\tilde{\mathcal{O}}(|A_{j}| (\sum_i |S_{i,j}|)^2/\epsilon^2 + m(j) - 1).
\end{align}
\end{lemma}
\begin{proof}
    We first prove that, before reaching the $\epsilon$-NE of the original game, in each windows $T_j,\ j<k-1$, the number of iterations of the last-window average strategy $|T_j|<\tilde{\mathcal{O}}(|A_{j}| (\sum_i |S_{i,j}|)^2/\epsilon^2 + m(j)-1)$ if the regret minimizer has $\tilde{\mathcal{O}}(\sqrt{|A_{i, j}}|S_{i, j}|/|T_j|)$ regret, where $|A_{j}|=\max_{i\in \mathcal{P}, s \in S}|A_{i,j}(s)|$. $A_{i,j}=\max_{s\in S_{i,j}}A_{i,j}(s)$ and $S_{i,j}$ are defined as usual, the set of actions and information sets, respectively. But the index $j$ here specify the time window. We prove this claim by contradiction:

\begin{claim}
    For $j =1,2,\cdots, k-1$, $|T_j|< \tilde{\mathcal{O}}(|A_{j}| (\sum_i |S_{i,j}|)^2/\epsilon^2 + m(j)-1)$.
\end{claim}
\begin{claimproof}
    Suppose the claim above is not correct, which is $|T_j|\geq\tilde{\mathcal{O}}(|A_{j}| (\sum_i |S_{i,j}|)^2/\epsilon^2 + m(j)-1)$, thus the average strategy in current window already reaches the restricted game's local $\epsilon$-NE. Then there must exist a BR computing right after $\tilde{\mathcal{O}}(|A_{j}|(\sum_i|S_{i,j}|)^2/\epsilon^2 + m(j)-1)$ iterations of regret minimization in current window. Suppose this iteration is the $t$-th one. Since $\mathbf{BR}(\bar{\pi}^t)\in \Pi_j$ and $\bar{\pi}^t$ has reached restricted game's local $\epsilon$-NE, $\bar{\pi}^t$ is also the $\epsilon$-NE in the original game, which cause contradiction since in $T_j$, $j<k$, the average strategy does not reach $\epsilon$-NE. Therefore the claim is correct.
\end{claimproof}

Then we prove the rest, for $j=k-1$, regret minimizer in the worst case needs $\tilde{\mathcal{O}}(|A_k| (\sum_i |S_{i,k}|)^2/\epsilon^2 + m(j)-1)$ iterations to reach $\epsilon$-NE. Since we have upper bound of number of iterations in each window before the last one. Training with a larger number of iterations will only approximate the NE closer. Thus we can sum up the required iterations in all time window, we get that in the worst case the required number of iterations for the LAS of RMDO to reach $\epsilon$-NE is:
\begin{align}
     \sum_{j=1}^k\tilde{\mathcal{O}}(|A_{j}| (\sum_i |S_{i,j}|)^2/\epsilon^2 + m(j) - 1).
\end{align}
\end{proof}

\begin{lemma}[Sample Complexity]
\label{lemma:sc}
The sample complexity of LAS in RMDO is:
\begin{align}
    \sum_{j=1}^k\tilde{\mathcal{O}}[|A_{j}| (\sum_i |S_{i,j}|)^3/\epsilon^2m(j) + &|A_{j}| (\sum_i |S_{i,j}|)^3/\epsilon^2 + \sum_i|S_{i,j}|m(j)]
\end{align}


\end{lemma}
\begin{proof}
In general, there are two parts contribute to the complexity, regret minimization and Best Response computing. We will compute the sample complexity of both separately.

We have proved the required iterations in Lemma \ref{lemma:ri}. Since the sample complexity of each iteration is $\sum_i|S_{i,j}|$, then we have the complexity of regret-minimization part is 
\begin{align}
     \sum_{j=1}^k\tilde{\mathcal{O}}(|A_{j}| (\sum_i |S_{i,j}|)^3/\epsilon^2 + \sum_i|S_{i,j}|m(j)-\sum_i|S_{i,j}|).
\end{align}

Then we compute the complexity in BR computing. In $T_j$, $j=1,2,\cdots,k$, we compute BR every $m(j)$ iterations of regret minimization, thus by multiplying total times of BR computation and the complexity in each computation, we can get the BR-computing part of sample complexity:
\begin{align}
    \sum_{j=1}^k\tilde{\mathcal{O}}(|A_{j}| (\sum_i |S_{i,j}|)^3/\epsilon^2m(j) + \sum_i|S_{i,j}| - \sum_i|S_{i,j}|/m(j)))
\end{align}

Sum up the above two parts of the complexity we can get the final complexity:
\begin{align}
\label{rmdo_sc}
    \sum_{j=1}^k\tilde{\mathcal{O}}(|A_{j}| (\sum_i |S_{i,j}|)^3/\epsilon^2m(j) &+ \sum_i|S_{i,j}| - \sum_i|S_{i,j}|/m(j)) \nonumber \\ 
    + \sum_{j=1}^k\tilde{\mathcal{O}}(|A_{j}| (\sum_i |S_{i,j}|)^3/\epsilon^2 &+ \sum_i|S_{i,j}|m(j)-\sum_i|S_{i,j}|), \nonumber \\
    \leq \sum_{j=1}^k\tilde{\mathcal{O}}[|A_{j}| (\sum_i |S_{i,j}|)^3/\epsilon^2m(j) &+ |A_{j}| (\sum_i |S_{i,j}|)^3/\epsilon^2 + \sum_i|S_{i,j}|m(j)]
\end{align}
We take the upper bound as the sample complexity, since if the number of sample points of the algorithm reach RHS of inequality (\ref{rmdo_sc}), we can guarantee that the LAS of RMDO reaches $\epsilon$-NE. This is common way of analyzing the theoretical sample complexity in game solving ~\citep{bai2022near}.
\end{proof}

\begin{customlemma}{\ref{lemma:x}}

$\min_{\pi\in \Pi^*}\max_{s\in S}\text{supp}^{\pi}(s) < k \leq \sum_i |S_{i,k}|\leq X \leq |S| =\sum_i |S_i|$.
\end{customlemma}

\begin{proof}
We proof from left to right:
\begin{itemize}
    \item Since there is at most one pure strategy at a infoset added to the population every time when the restricted game is expanded and a new window starts. Thus at $s\in S$, $|\{a|a\in A(s)\}| \leq k$. Since at every infoset, the number of pure strategies in the converged population will be greater than the support of NE. Otherwise, there is a pure strategy in the NE strategy but not in the population, which means that the population doesn't converge and leads to contradiction. At $s\in S$, denote $\Pi^*$ as a set of NE of this game, $|\{a|a\in A(s)\}| \geq \min_{\pi\in \Pi^*}\text{supp}^{\pi}(s)$. So we have $k\geq \max_{s\in S}|\{a|a\in A(s)\}| \geq \max_{s\in S}\min_{\pi\in \Pi^*}\text{supp}^{\pi}(s)$.
    \item The number of restricted games is less than or equal to the number of infosets in the largest possible restricted game since at least one new action added, i.e. $\forall j=1,\cdots,k-1$, $\sum_i |S_{i,j+1}| - \sum_i |S_{i,j}| \geq 1$. Therefore, $\sum_i |S_{i,k}| \geq k$.
    \item It's easy to see that $X$ is the number of infosets in the largest possible restricted game constructed by NE. Thus, for any restricted game, the number of the infosets is smaller than or equal to $X$.
    \item And since such game is constructed by the actions from the original game, it is less than the number of the infosets in the original games.
\end{itemize}
\end{proof}

\begin{customthm}{\ref{theo:rmdo_sc}}
The sample complexity of last-window average strategy of RMDO to reach $\epsilon$-NE is:
\begin{align}
\tilde{\mathcal{O}}( k|A|X^3/\epsilon^2 + \sum_{j=1}^k|A|X^3/\epsilon^2m(j) +  Xm(j))
\end{align}
\end{customthm}

\begin{proof}
According to Lemma \ref{lemma:x}, we have $\forall j=1,2,\cdots,k,\ |A_j| \leq |A|$ and $\sum_i |S_{i,j}|\leq X$. Then according to Lemma \ref{lemma:sc}, the sample complexity to reach $\epsilon$-NE is 
\begin{align}
\sum_{j=1}^k\tilde{\mathcal{O}}\Big(|A_{j}| (\sum_i |S_{i,j}|)^3/\epsilon^2m(j) + |A_{j}| (\sum_i |S_{i,j}|)^3/\epsilon^2 + \sum_i|S_{i,j}|m(j)\Big) \nonumber \\
\leq \tilde{\mathcal{O}}( k|A|X^3/\epsilon^2 + \sum_{j=1}^k|A|X^3/\epsilon^2m(j) +  Xm(j)).
\end{align}
For easier comparison between algorithms and without loss of generality, we take the upper bound as the sample complexity.
\end{proof}

\begin{customthm}{\ref{theo:xodo_bound}}
    The sample complexity for XODO to reach $\epsilon$-NE is $\tilde{\mathcal{O}}(2X^3k^2/\epsilon^2)$.
\end{customthm}

\begin{proof}
Setting $m(j)=1$ in Theorem \ref{theo:RMDO_oas_bound}, the upper bound will be $\tilde{\mathcal{O}}(|S_{i,k}|\sqrt{k}\sum_j\sqrt{|T_j|}/T)$. According to Cauchy-Schwartz inequality, $\sum_j\sqrt{|T_j|}\leq \sqrt{k\sum_j |T_j|}=\sqrt{kT}$, then the upper bound becomes $\tilde{\mathcal{O}}(|S_{i,k}|k/\sqrt{T})$. Thus the required iteration $T$ to reach $\epsilon$-NE satisfies:
\begin{align}
    \sum_i \tilde{\mathcal{O}}(|S_{i,k}|k/\sqrt{T}) \leq \epsilon.
\end{align}
For clarity and easy for comparison between games, as usual we let the upper bound of LHS satisfy the above condition, which means:
\begin{align}
    T \geq \tilde{\mathcal{O}}(k^2X^2/\epsilon^2)
\end{align}
\end{proof}



\begin{customthm}{\ref{theo:SRMDO_ebr}}
The Last-window average strategy of SRMDO has the following sample complexity to reach $\epsilon$-Nash Equilibrium when employing \textbf{oracle Best Response}.
\begin{align}
    \tilde{\mathcal{O}}( kH|A|X^2/\epsilon^2 + \sum_{j=1}^k H_{j} m(j) +   |A|X^3/\epsilon^2m(j)),
\end{align}
and the following sample complexity when employing $\delta$-BR:
\begin{align}
    \tilde{\mathcal{O}}( kH|A|X^2/(\epsilon - \delta)^2 + \sum_{j=1}^k H_{j} m(j) +   H|A|X^2/(\epsilon - \delta)^2m(j))
\end{align}
, where it is required that $\delta<\epsilon$, $H_{j} = \max_i H_{i,j}$ is the largest horizon of the restricted game in $T_j$, and obviously $H_{j} \ll |S_{i,j}|$, typically $H_{j} = \tilde{\mathcal{O}}(\log(|S_{i,j}|))$.
\end{customthm}

\begin{proof}
Similarly, before reaching the $\epsilon$-NE of the original game, in each windows $T_j,\ j<k-1$, the number of iterations of the last-window average strategy $|T_j|<\tilde{\mathcal{O}}(|A_{j}| (\sum_i |S_{i,j}|)^2/\epsilon^2 + m(j)-1)$ if the regret minimizer has $\tilde{\mathcal{O}}(\sqrt{|A_{j}}|S_{i, j}|/|T_j|)$ regret. Moreover, regret minimizer in the worst case needs $\tilde{\mathcal{O}}(|A_k| (\sum_i |S_{i,k}|)^2/\epsilon^2 + m(j)-1)$ iterations to reach $\epsilon$-NE. However, when computing the sample complexity, unlike RMDO, the complexity of each iteration reduce from $\sum_i|S_{i,j}|$ to $H_{j}$, where $H_{j}$ indicates the longest length of the trajectories in the interactiosn between stochastic regret minimizer and the environments.

Then we get the complexity of compute Best Responses. The number of times computing BR in each window $T_j$ is $\tilde{\mathcal{O}}(|A_{j}| (\sum_i |S_{i,j}|)^2/\epsilon^2m(j) + 1 - 1/m(j))$, and the complexity of each round of BR computing is $\tilde{\mathcal{O}}(\sum_i|S_{i,j}|)$.

Then we collect all parts of the complexity to get:
\begin{align}
    \sum_{j=1}^k\tilde{\mathcal{O}}(|A_{j}| (\sum_i |S_{i,j}|)^3/\epsilon^2m(j) + \sum_i|S_{i,j}| - \sum_i|S_{i,j}|/m(j)) \nonumber \\ 
    + \sum_{j=1}^k\tilde{\mathcal{O}}(|A_{j}| (\sum_i |S_{i,j}|)^2H_{j}/\epsilon^2 + H_{j}m(j)-H_{j}), \nonumber \\
    \leq \sum_{j=1}^k\tilde{\mathcal{O}}[|A_{j}| (\sum_i |S_{i,j}|)^3/\epsilon^2m(j) + |A_j| (\sum_i |S_{i,j}|)^2H_{j}/\epsilon^2 + H_{j}m(j)]  \label{eq:real_srmdo_ada_sc}\\
    \leq \tilde{\mathcal{O}}( kH|A|X^2/\epsilon^2 + \sum_{j=1}^k H m(j) +   |A|X^3/\epsilon^2m(j))
\end{align}

As for the $\delta$-BR approximate Best Responder, it's similar to the proof above, except in each window, $|T_j|<\tilde{\mathcal{O}}(|A_{j}| (\sum_i |S_{i,j}|)^2/(\epsilon - \delta)^2 + m(j)-1)$ to ensure the following: we still can derive contraction via the fact that $\delta$-BR of the average strategy ($\epsilon-\delta$-NE) in current window still lies in current restricted game leading to $\epsilon$-NE in the original game. Additionally, approximate BR only have complexity $\tilde{\mathcal{O}}(H)$ in each time of computation. Thus it's easy to derive the complexity:
\begin{align}
    \tilde{\mathcal{O}}( kH|A|X^2/(\epsilon - \delta)^2 + \sum_{j=1}^k H_{j} m(j) +   H|A|X^2/(\epsilon - \delta)^2m(j))
\end{align}
\end{proof}

\begin{customthm}{\ref{theo:srmdo_ada}}
By employing exact Best Response and the following frequency function
\begin{align}
    m(j) = \sqrt{\frac{|A_j|(\sum_i|S_{i,j}|)^3}{H_{j}\epsilon^2}},
\end{align}
the sample complexity of Last-window average strategy in SRMDO is
\begin{align}
\tilde{\mathcal{O}}( k|A|HX^2/\epsilon^2 + 2k\sqrt{|A|H}X^{1.5}/\epsilon).
\end{align}
\end{customthm}

\begin{proof}
According to the Proof to Theorem \ref{theo:SRMDO_ebr}, we have in equation (\ref{eq:real_srmdo_ada_sc}), the sample complexity of SRMDO is:
\begin{align}
\sum_{j=1}^k\tilde{\mathcal{O}}[|A_{j}| (\sum_i |S_{i,j}|)^3/\epsilon^2m(j) + |A_j| (\sum_i |S_{i,j}|)^2H_{j}/\epsilon^2 + H_{j}m(j)],
\end{align}
which can be proved with the similar idea in proof to Theorem \ref{theo:adado_sc} that has the following lower bound:
\begin{align}
\sum_{j=1}^k\tilde{\mathcal{O}}[2\sqrt{H_j|A_{j}|} (\sum_i |S_{i,j}|)^{1.5}/\epsilon + |A_j| (\sum_i |S_{i,j}|)^2H_{j}/\epsilon^2]
\label{eq:real_ado_sc}
\end{align}
when
\begin{align}
m(j) = \sqrt{\frac{|A_j|(\sum_i|S_{i,j}|)^3}{ H_{j}\epsilon^2}}.
\end{align}
As usual, we treat the tight upper bound of equation (\ref{eq:real_ado_sc}) as the sample complexity, which is
\begin{align}
\tilde{\mathcal{O}}( k|A|HX^2/\epsilon^2 + 2k\sqrt{|A|H}X^{1.5}/\epsilon).
\end{align}
\end{proof}

\section{Games Descriptions}
\label{append:games}
\paragraph{Blotto} Sequential perfect-information game \citep{tang2023regret} is a revised sequential version of discrete Colonel Blotto Game. At the outset, each participant possesses a distinct array of forces varying in strength, which they sequentially deploy onto the battlefield for engagement. Precisely, under the parameter configuration $20$, each combatant commands $20$ forces, with strengths ranging from $0$ to $19$, and engages in $4$ successive deployments, thus constituting two complete combat cycles. During each deployment iteration, players alternately select forces for deployment. The outcome of each fight is the difference in strengths between the forces present on the battlefield. Upon conclusion of a round, the deployed forces are removed, paving the way for the commencement of the subsequent round. The payoff of the game is the summation of the outcomes of all rounds.
\paragraph{Large Kuhn Poker} It is a variant of Kuhn Poker created by ~\citet{tang2023regret}. It is only different from normal Kuhn Poker with an initial pot for each player $40$. Players can bet any remaining amount.

\paragraph{Leduc Poker 10 Card} It is the same as vanilla Leduc Poker except the initial number of cards is $10$.

\paragraph{Leduc Poker Dummy} It is the same as vanilla Leduc Poker except the actions in each information set are duplicated once~\citep{xdo}.

\paragraph{Oshi zumo} It is a board game ~\citep{buro2004solving}, where two players have 4 coins in the beginning of the game. There is a token in the middle of a board with length $2K$ + 1. $K$ in our case is 6. Players make action of bidding with the coins they have (at least $1$). Then the player bid more is able to push the token one step toward its opponent. The objective is to push the token off the opponent side of the board. Payoff is $1$ for the winner and $-1$ for the loser.

\section{Additional Experimental Results}

We offer additional exploitability plots of the experimental results in this section.
\begin{figure}[h]
    \centering
    \includegraphics[width=0.9\textwidth]{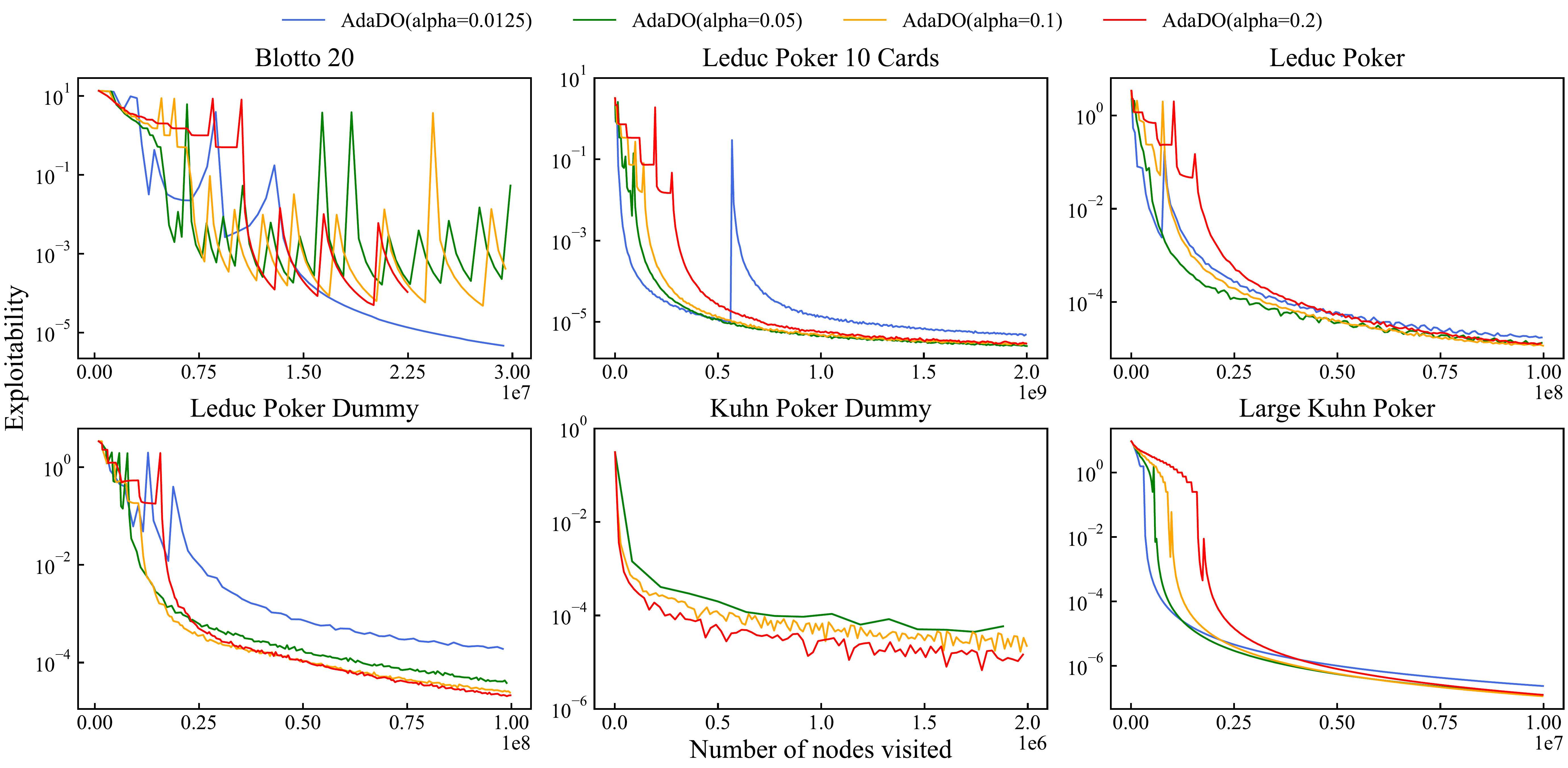}

    \caption{Ablation study of the discount factor $\alpha$ in the practical AdaDO.}
\end{figure}

\begin{figure}[h]
    \centering
    \includegraphics[width=0.9\textwidth]{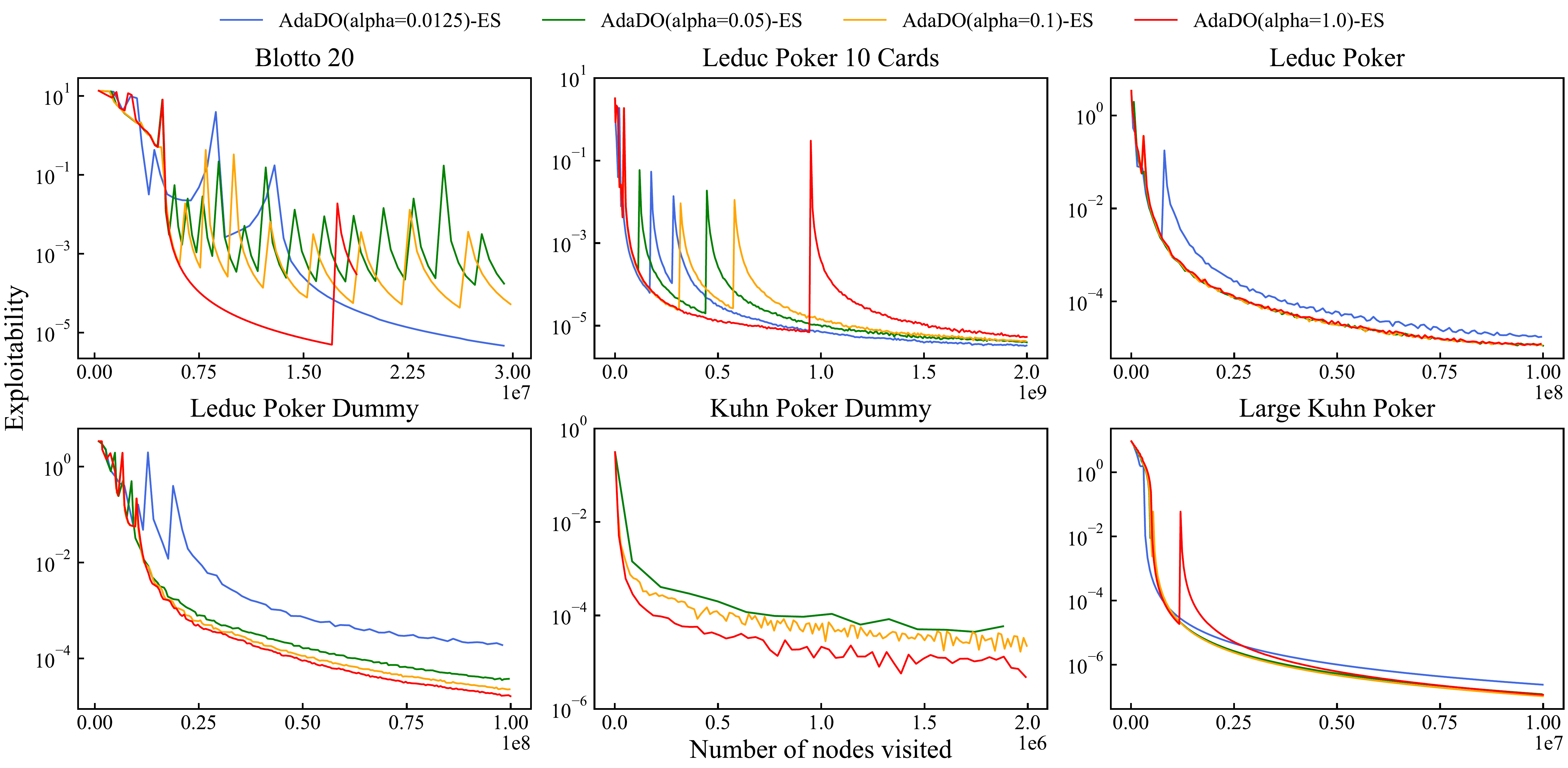}

    \caption{Ablation study of the discount factor $\alpha$ in the practical AdaDO with early stopping trick.}
\end{figure}

\vskip 0.2in
\newpage
\bibliography{sample}

\end{document}